\documentclass[aoas,preprint]{imsart}
\setattribute{journal}{name}{}
\RequirePackage[round]{natbib}
\RequirePackage[colorlinks,citecolor=blue,urlcolor=blue]{hyperref}

\usepackage{bbm}
\usepackage{comment}
\usepackage{amsmath, amsthm, amssymb}
\usepackage{algorithm, algorithmicx}
\usepackage{graphicx}
\usepackage{subfigure}
\usepackage{lineno}
\usepackage[usenames,dvipsnames]{color}
\usepackage{multirow}
\usepackage{booktabs}

\let\hat\widehat
\newtheorem{thm}{Theorem}

\theoremstyle{remark}
\newtheorem{remark}{Remark}

\newcommand\R{\mathbb{R}}
\newcommand\E{\mathbb{E}}
\renewcommand\P{\mathbb{P}}

\newcommand\cE{{\cal E}}

\catcode`@=11
\newskip\beforeproofvskip
\newskip\afterproofvskip
\beforeproofvskip=\medskipamount
\afterproofvskip=\bigskipamount

\def\prooftag{Proof}
\def\proofskip{\enspace}

\def\proof{\@ifnextchar[{\@@proof}{\@proof}}  
\def\@startproof{\par\vskip\beforeproofvskip\leavevmode}
\def\@proof{\@startproof{\scshape\prooftag.}\proofskip}
\def\@@proof[#1]{\@startproof {\scshape\prooftag #1.}\proofskip}

\catcode`@=12

\let\hat\widehat
\let\tilde\widetilde

\begin{document}

\begin{frontmatter}

\title{On the use of bootstrap with variational inference: Theory, interpretation, and a two-sample test example}
\runtitle{Bootstrapping the Variational Inference}

\begin{aug}
  \author{\fnms{Yen-Chi}
    \snm{Chen}\ead[label=e1]{yenchic@uw.edu}},
  \author{\fnms{Y. Samuel}
    \snm{Wang}\ead[label=e2]{ysamwang@uw.edu}},
    \and
  \author{\fnms{Elena A.}
    \snm{Erosheva}\ead[label=e3]{erosheva@uw.edu}}
  \affiliation{Department of Statistics\\University of Washington}
  \runauthor{Y.-C. Chen et al.}
  \address{Department of Statistics\\University of Washington\\
Box 354322\\ Seattle, WA 98195 \\
          \printead{e1}}
    \address{Department of Statistics,\\ University of Washington\\
Box 354322\\ Seattle, WA 98195 \\
          \printead{e2}}
    \address{Department of Statistics,\\ School of Social Work,\\ and Center for Statistics and the Social Sciences\\University of Washington\\
Box 354322\\ Seattle, WA 98195 \\
          \printead{e3}}    
\end{aug}

\begin{abstract}
Variational inference is a general approach for approximating complex density functions, such as those arising in latent variable models, popular in machine learning. It has been applied to approximate the maximum likelihood estimator and to carry out
Bayesian inference, however, quantification of uncertainty with variational inference remains challenging from both theoretical and practical perspectives. This paper is concerned with developing uncertainty measures for variational inference by using bootstrap procedures. We first develop two general bootstrap approaches for assessing the uncertainty of a variational estimate
and the study the underlying bootstrap theory in both fixed- and increasing-dimension settings. We then use the bootstrap approach and our theoretical results in the context of mixed membership modeling with multivariate binary data on functional disability from the National Long Term Care Survey. We carry out a two-sample approach to test for changes in the repeated measures of functional disability for the subset of individuals present in 1989 and 1994 waves.
\end{abstract}

\end{frontmatter}

\section{Introduction}

Variational inference \citep{jordan1999introduction,wainwright2008graphical} is a method to approximate
complex density functions \citep{blei2017variational} which has been applied to various statistical models such as
factor analysis \citep{ghahramani2000variational,khan2010variational,klami2015group},
stochastic block models \citep{celisse2012consistency,latouche2012variational, bickel2013asymptotic},
latent Dirichlet allocation \citep{blei2003latent,blei2006variational}, and
Gaussian processes \citep{damianou2011variational,damianou2014variational}.

Variational inference can be used to approximate
a posterior distribution as an alternative to Markov Chain Monte Carlo (MCMC), when a sampling procedure would be prohibitively slow or require immense human effort to tune,
or to approximate
a maximum likelihood estimator (MLE), when computation with the specified likelihood is intractable.
In particular, when the model involves a latent structure such
as a mixed membership model \citep{airoldi2005latent, airoldi2008mixed,wang2015variational}
or a mixed effect model \citep{hall2011theory, westling2015establishing},
finding the MLE may be very challenging
and variational inference provides
a fast way to obtain an estimate of the parameter.
The estimator from variational inference is called the variational estimator.





Recently,
the asymptotic distribution
of point estimates resulting from variational inference was investigated in
\cite{hall2011asymptotic,bickel2013asymptotic,westling2015establishing}
and \cite{wang2017frequentist} analyze variational inference under a Bayesian framework.
When a consistent estimator of the asymptotic variance is available, practitioners can analyze the uncertainty of the variational estimate and draw scientific conclusions
by constructing confidence intervals (CI) for the parameters of interest 
\citep{hall2011asymptotic, westling2015establishing}.

However, constructing a CI using the asymptotic distribution
fails if we do not have a consistent estimator of the variance of the variational estimator.
To overcome this problem,
we consider using bootstrap methods implemented in
\cite{bickel2013asymptotic} and \cite{wang2015variational}.
The bootstrap approach does not require a consistent variance estimator to be available,
and, in some cases, leads to a CI with a higher-order coverage \citep{hall2013bootstrap}.
Despite the fact that the bootstrap method has already been used with variational estimation~\citep{wang2015variational},
the underlying bootstrap theory for variational inference does not exist. 


%

In this paper,
we investigate the validity of using a bootstrap approach where variational inference is used to approximate an MLE.
We construct a confidence interval (CI) in the usual fixed dimensional case, where
both the dimensionality of the parameter and the number of latent variables are fixed,
as well as in the increasing dimensional case. An example of the latter situation may come from an item response theory model where the latent dimensionality may increase when the number of questions per individual is increasing.
\cite{haberman1977maximum} and \cite{douglas1997joint} have analyzed
a situation where the number of questions (dimension)
and the number of participants (sample size) increase jointly.


This paper has been motivated by the general need -- as opposed to one specific substantive problem or a specific application area -- to provide statisticians, computer scientists and data scientists with the theory and tools for using the bootstrap for variational inference. We use two sets of functional disability measures obtained five years apart from the National Long Term Care Survey (NLTCS) to illustrate the bootstrap approach on a two-sample test, a setting where we find the variational inference to be particularly appropriate. However, a complete development of a substantive application is beyond the scope of this paper. 


\paragraph{Outline}
We briefly review variational inference in Section~\ref{sec::VI}.
In Section~\ref{sec::BT}, we discuss how to apply the bootstrap
to variational inference.
We then develop asymptotic normality and bootstrap theory in
Section~\ref{sec::theory}.
In Section~\ref{sec::NLTCS}, we illustrate the bootstrap approach with a two-sample test using functional disability data from the NLTCS.
Finally, we discuss related topics and the link to Stephen E. Fienberg in Section~\ref{sec::discussion}.

\section{Variational Inference}	\label{sec::VI}

We consider the variational inference in the context of a latent variable model.
Assume our data consists of $n$ individuals and $J$ variables (e.g., survey questions or test problems)
and forms a random sample of $X_1,\cdots,X_n \in \R^J$ that are IID from a
distribution function $P_0$.
We assume that there exists $K$ latent features for each individual that are denoted as
$Z_1,\cdots,Z_n \in \R^K$.
This setup is quite general -- in a mixed membership model, $Z_i$
is the vector of mixed membership indicator;
in a random effect model, $Z_i $ is the random effect;
in a stochastic block model, $Z_i$ is the community indicator
(and $X_i = \{0,1\}^n$ denotes the edge connected to the $i$-th observation).

We assume a parametric model on the distribution function $P_0$
such that the joint distribution of $(X_i,Z_i)$
has a parametric form
$P(x,z;\theta)$, where $\theta\in\Theta\subset \R^d$ is the parameter of interest.
When the latent feature vector $Z_i$
is known, the likelihood of $i$-th observation is
$$
L(\theta|X_i,Z_i) = P(X_i, Z_i;\theta).
$$
Analogously to \cite{neyman1948consistent},
we regard the latent feature vectors $Z_1,\cdots,Z_n$ as
incidental parameters and the population parameter $\theta$ as structural parameters.


In reality, we do not know the latent vectors so the observed log-likelihood function is
\begin{equation}
\ell(\theta|\mathcal{X}) = \log L(\theta|X_i) =\log \int P(X_i,Z_i;\theta)dZ_i.
\label{eq::likelihood}
\end{equation}
Often, we are interested in using
the maximum likelihood estimator
$$
\hat{\theta}_{MLE} = \underset{\theta}{\sf argmax}\,\, \sum_{i=1}^n\ell(\theta|X_i).
$$

However, maximizing or even calculating the marginal likelihood can often be computationally intractable. Thus,
variational estimators provide an alternative, computationally feasible estimator
to the MLE.
The variational estimator is constructed as follows.
We first pick a family of distributions--the variational distribution family--for
the latent variable $Z_i$.
Let $Q(z;\omega)$ be the variational distribution family indexed by the variational parameter
$\omega\in\Omega \subset\R^s$, which is
a nuisance parameter in our model.
Note that we allow each $Z_i$ has its own variational parameter; namely, $\omega=\omega_i$. 
Using Jensen's inequality, the log-likelihood function satisfies
\begin{equation}
\begin{aligned}
\ell(\theta|X_i) &= \log\int P(X_i,Z_i;\theta)dZ_i\\
& = \log\int \frac{P(X_i,Z_i;\theta)}{Q(Z_i;\omega_i)}Q(Z_i;\omega_i)dZ_i\\
& = \log\E_{Z_i\sim Q}\left(\frac{P(X_i,Z_i;\theta)}{Q(Z_i;\omega_i)}|X_i\right)\\
& \geq \E_{Z_i\sim Q}\left(\log P(X_i,Z_i;\theta)|X_i\right) - \E_{Z_i\sim Q}\left(\log Q(Z_i;\omega_i)\right)\\
& = {\sf ELBO}(\theta,\omega_i|X_i),
\end{aligned}
\end{equation}
where $\E_{Z_i\sim Q_i}$ means that the expectation is taken over variable $Z_i$
and the underlying distribution is $Q(\cdot;\omega_i)$.
We call the expression
on the right hand side of the inequality the \emph{evidence lower bound} (ELBO).

Instead of maximizing the log-likelihood function, the variational framework
maximizes the ELBO, leading to
\begin{equation}
\hat{\theta}_{ELBO}, \hat{\omega}_{ELBO,1},\cdots,  \hat{\omega}_{ELBO,n} 
= \underset{\theta,\omega_1,\cdots,\omega_n}{\sf argmax} \,\, \sum_{i=1}^n {\sf ELBO}(\theta,\omega_i|X_i).
\label{eq::elbo0}
\end{equation}
Because $\omega_i$ in the above maximizing criterion is only involved in ${\sf ELBO}(\theta,\omega_i|X_i)$ when $\theta$ is fixed,
the first element $\hat{\theta}_{ELBO}$ is equivalent to the maximizer of the following criterion:
\begin{equation}
\begin{aligned}
\hat{\theta}_{ELBO}
&= \underset{\theta}{\sf argmax} \,\, \sum_{i=1}^n {\sf ELBO}(\theta,\omega_{\max}(\theta|X_i)|X_i)\\
&= \underset{\theta}{\sf argmax} \,\, \sum_{i=1}^n \cE(\theta|X_i),
\end{aligned}
\label{eq::elbo}
\end{equation}
where $\omega_{\max}(\theta|X_i) = \underset{\omega_i}{\sf argmax}\,\, {\sf ELBO}(\theta,\omega_i|X_i).$
The quantity $\hat{\theta}_{ELBO}$,
is called the ELBO estimator
or the variational estimator.

Because the ELBO estimator comes from optimizing $\sum_{i=1}^n \cE(\theta|X_i)$, it is an estimator of
\begin{equation}
\theta_{ELBO}= \underset{\theta}{\sf argmax} \,\, \E \left(\cE(\theta|X_1)\right).
\label{eq::elbo_pop}
\end{equation}
Note that the expectation in the above expression is for the random variable $X_1$
and is taken with respect to the data-generating distribution $P_0$.
The quantity $\theta_{ELBO}$ defines the population quantity
that the variational inference (ELBO estimator) is estimating.
Note that $\theta_{ELBO}$ depends on the variational distribution $Q$
and is often different from 
the population version of 
$\theta_{MLE} = \underset{\theta}{\sf argmax} \,\, \E \left(\ell(\theta|X_1)\right)$.
Thus,
variational inference can be thought of as an intentional model misspecification
even if the original parametric model is correctly specified.
We will argue in the next section that despite the misspecification, variational inference is still a useful procedure
for making statistical inference.


\begin{figure}[htb]
\includegraphics[height=1.5in]{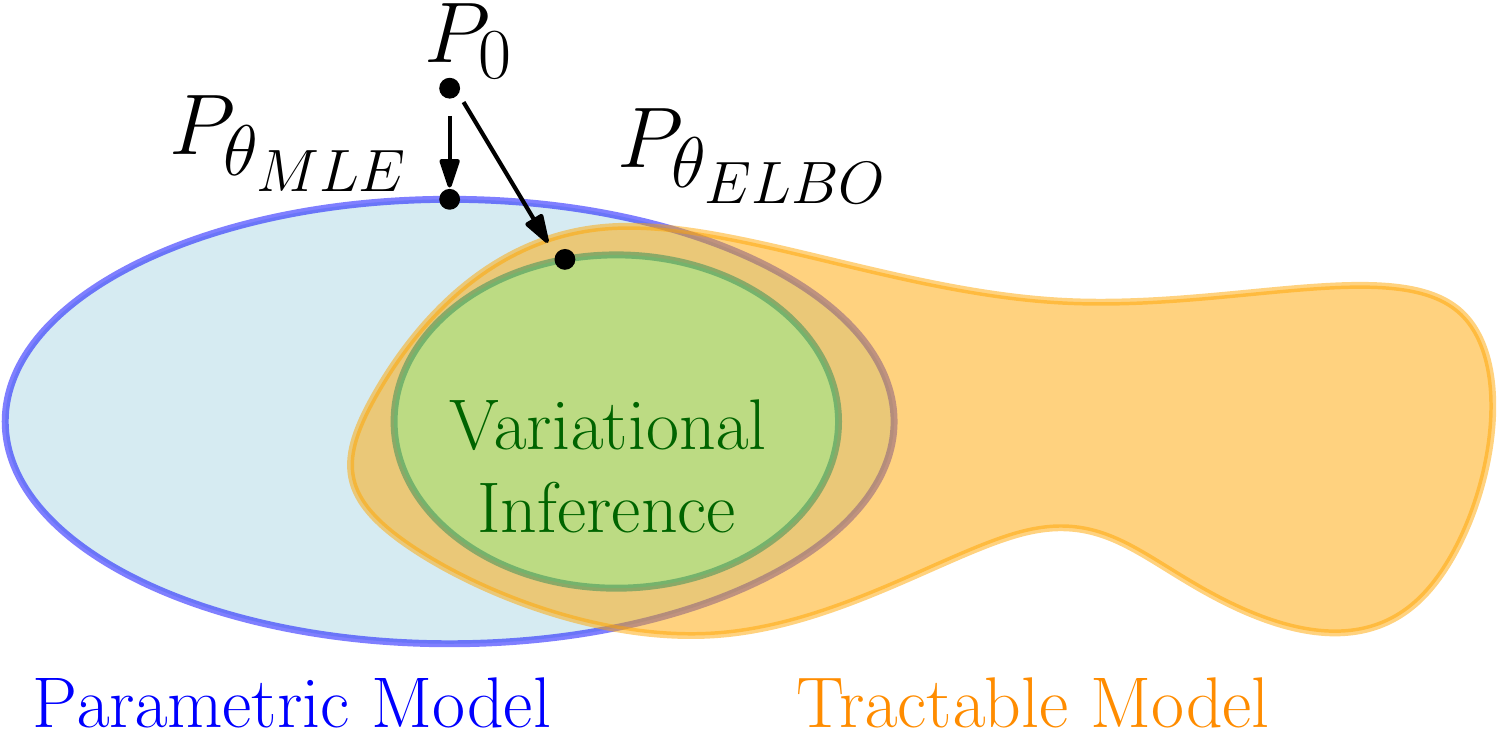}
\caption{An illustration for the relations of $P_0, P_{\theta_{MLE}},$ and $P_{\theta_{ELBO}}$
when a latent variable model is used and the model is not correctly specified.
In this case, the distribution corresponding to the population MLE
is just the distribution in the parametric family that minimizes the KL-divergence
to the true distribution function.
So $P_{\theta_{MLE}}$ can be viewed as a projection from $P_0$ onto the parametric family.
However, if the MLE is computationally intractable, we can still specify a tractable variational estimator and the corresponding variational distribution, $P_{\theta_{ELBO}}$,
can be viewed as another projection from $P_0$.
}
\label{fig::map}
\end{figure}


\begin{remark}
When the parametric model is correctly specified (i.e., there exists $\theta_0\in\Theta$ such that $P_0=P_{\theta_0}$),
the variational estimator may recover the correct model with $\theta_{ELBO} = \theta_0$
in some special cases.
For concrete examples, we refer the readers to
\cite{hall2011asymptotic}
and \cite{bickel2013asymptotic}
where they illustrated this possibility in
a single predictor Poisson mixed model
and a stochastic block model.
\end{remark}

\subsection{Further considerations for using variational inference in practice}\label{sec:furtherConsid}

As described in the previous section,
the distribution based on the variational estimator $P_{\hat{\theta}_{ELBO}}$
may not converge to the true data-generating distribution function even
when the model is correctly specified.
%
Despite this drawback, variational inference can be a useful procedure
for inference for the following reasons.

%

\begin{itemize}
\item {\bf Likelihood formulation as a working model.}
As George Box has said \emph{``Essentially, all models are wrong, but some are useful"}
\citep{doi:10.1080/01621459.1976.10480949}.
A proposed model is almost always misspecified.
When using a parametric model to analyze the data,
we do not claim that the parametric model describes the actual data-generating distribution.
Instead,
a working model and parameter estimates help us to learn various aspects about the data at hand.
To carry this reasoning further, the ML procedure and the variational inference procedure
are just different principles of fitting parameters to the data.
When the model is misspecified,
both the MLE and the variational estimator
are best approximation estimators under different criteria
of measuring the quality of approximation.
Figure~\ref{fig::map}
provides a diagram illustrating the case where the likelihood model
does not include the true data-generating distribution.


\item {\bf Two-sample test.}
The variational inference procedure
is a useful procedure for two-sample test.
Given two sets of data $X_1,\cdots,X_n\sim P_X$
and $Y_1,\cdots, Y_m \sim P_Y$,
the goal of a two-sample test is to test
$$
H_0: P_X = P_Y.
$$
Using the equation \eqref{eq::elbo_pop},
$H_0$ implies that
$$
\theta_{ELBO,X} = \theta_{ELBO,Y},
$$
where $\theta_{ELBO,X}$ and $\theta_{ELBO,Y}$ are the maximizer of equation \eqref{eq::elbo_pop}
assuming the expectation is taken over $P_X$ and $P_Y$.
Thus, when applied to a two-sample test,
variational inference is as valid of an approach as ML inference.
In a sense, one can interpret the tests following either approach, variational or ML, as inferences based on different projections of the distributions $P_X,P_Y$ onto the same parameter space.

\end{itemize}


\section{Bootstrapping the variational estimator}	\label{sec::BT}


We use the bootstrap \citep{efron1982jackknife,efron1992bootstrap} to
evaluate the uncertainty of the variational estimator and construct CIs.
We focus on the empirical bootstrap -- also known as classical, nonparametric, or Efron's bootstrap --  where
one samples with replacement from the original dataset, recomputes the ELBO estimator for each bootstrap sample,
and uses the distribution of these bootstrapped ELBO estimators to derive uncertainty measures.
We illustrate estimation of the error of $\hat{\theta}_{ELBO}$
and construction of the CIs using the bootstrap.
There are many bootstrap CIs (see, e.g., \citealt{hall2013bootstrap}).
Here, we will focus on two common approaches:
the percentile method and the (studentized) pivotal method.
Note that constructing a CI using the percentile approach has been implemented in \cite{wang2015variational}.

The bootstrap approach to estimating uncertainty is very general.
The bootstrap percentile approach can be used even when the asymptotic covariance matrix is not available (e.g., difficult to estimate).
When the asymptotic covariance matrix of $\hat{\theta}_{ELBO}$ is known and
can be consistently estimated (say using a sandwich estimator),
the bootstrap pivotal method 
may produce CIs with a higher order coverage than those based on the asymptotic normality
\citep{horowitz1997bootstrap,singh1981asymptotic,babu1983inference}.

More formally, let
$X^*_1,\cdots,X^*_n$ be a bootstrap sample from the original sample $\mathcal{X} = \{X_1,\cdots,X_n\}$.
Given the bootstrap sample, we compute the bootstrap ELBO estimator
\begin{equation}
\hat{\theta}^*_{ELBO}= \underset{\theta}{\sf argmax} \,\, \sum_{i=1}^n \cE(\theta|X^*_i).
\label{eq::elbo_bt}
\end{equation}
Repeating the bootstrap procedure $B$ times, we obtain $B$ bootstrap ELBO estimators:
$$
\hat{\theta}^{*(1)}_{ELBO},\cdots,\hat{\theta}^{*(B)}_{ELBO}.
$$
We will use these bootstrap ELBO estimators
to assess the uncertainty of the original ELBO estimator.


Note that one may also use the
jackknife and weighted bootstrap \citep{o2015estimation}
to generate bootstrap sample. 
In particular, when analyzing a network data where each $X_i$
corresponds to the edges of $i$-the vertex,
the empirical bootstrap cannot be applied
but the weighted bootstrap (and the parametric bootstrap; see Remark~\ref{rm::parametric})
is still applicable.

\subsection{Estimating the variance}
The bootstrap approach can be applied to estimate the variance of
the ELBO estimator.
Assume that we focus on the $\ell$-th parameter $\theta_\ell$.
The variance of $\hat{\theta}_{ELBO,\ell}$ can be estimated using
the sample variance of the bootstrapped variational estimators
\begin{equation}
\hat{\sf Var}(\hat{\theta}_{ELBO, \ell}) = \frac{1}{B}\sum_{j=1}^B\left(\hat{\theta}^{*(j)}_{ELBO,\ell}-\bar{\theta}^{*}_{ELBO,\ell}\right)^2, \quad
\bar{\theta}^{*}_{ELBO,\ell} = \frac{1}{B}\sum_{j=1}^B \hat{\theta}^{*(j)}_{ELBO,\ell}.
\label{eq::var_bt}
\end{equation}
Figure~\ref{fig::alg::var} provides a diagram summarizing the procedure.

The intuition behind equation \eqref{eq::var_bt} is that
the bootstrap distribution of the estimators $\hat{\theta}^{*(1)}_{ELBO},\cdots,\hat{\theta}^{*(B)}_{ELBO}$
behaves as if new realizations of the original estimator $\hat{\theta}_{ELBO}$ are drawn.
Thus, the variance of the bootstrap estimators would be an approximation to
the variance of $\hat{\theta}_{ELBO}$.

\begin{figure}[htb]
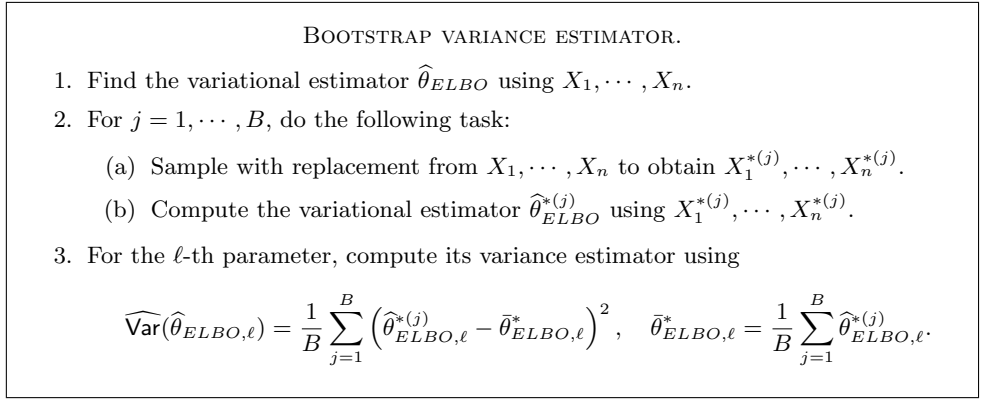

\fbox{\parbox{5in}{
\begin{center}
{\sc Bootstrap variance estimator.}
\end{center}
\begin{center}
\begin{enumerate}
\item Find the variational estimator $\hat{\theta}_{ELBO}$ using $X_1,\cdots,X_n$.
\item For $j=1,\cdots, B$, do the following task:
	\begin{enumerate}
	\item Sample with replacement from $X_1,\cdots,X_n$ to obtain
	$X^{*(j)}_1,\cdots,X^{*(j)}_n$.
	\item Compute the variational estimator $\hat{\theta}^{*(j)}_{ELBO}$ using $X^{*(j)}_1,\cdots,X^{*(j)}_n$.
	\end{enumerate}
\item For the $\ell$-th parameter,
compute its variance estimator using
$$
\hat{\sf Var}(\hat{\theta}_{ELBO, \ell}) = \frac{1}{B}\sum_{j=1}^B\left(\hat{\theta}^{*(j)}_{ELBO,\ell}-\bar{\theta}^{*}_{ELBO,\ell}\right)^2, \quad
\bar{\theta}^{*}_{ELBO,\ell} = \frac{1}{B}\sum_{j=1}^B \hat{\theta}^{*(j)}_{ELBO,\ell}.
$$
\end{enumerate}
\end{center}
}}
\caption{Bootstrap variance estimator.}
\label{fig::alg::var}
\end{figure}

\subsection{Confidence interval: percentile approach}

The bootstrap approach enables us to construct CIs for the parameters of interest.
We first introduce a simple approach called the percentile (quantile) approach,
which is based on the percentile of the distribution of the bootstrap variational estimators.
Assume again we focus on the $\ell$-th parameter.
Given a confidence level $\alpha$, let
$\hat{s}_{\ell,\alpha}$ denotes the $\alpha$-quantile of the bootstrap ELBO estimators
$$
\hat{s}_{\ell,\alpha} =  \hat{G}_\ell^{-1}(1-\alpha), \quad \hat{G}_\ell(t) = \frac{1}{B}\sum_{j=1}^B I(\hat{\theta}^{*(j)}_{ELBO, \ell}-\hat{\theta}_{ELBO,\ell}\leq t).
$$
Then a $(1-\alpha)$ CI of the $\ell$-th parameter is
\begin{equation}
C_{n,\alpha,\ell} = \left[\hat{\theta}_{ELBO,\ell}+\hat{s}_{\ell,\alpha/2},\,\, \hat{\theta}_{ELBO,\ell}+\hat{s}_{\ell,1-\alpha/2}\right].
\label{eq::CI_bt}
\end{equation}
Figure~\ref{fig::alg::PA} summarizes the steps in computing a bootstrap percentile CI.

Equation \eqref{eq::CI_bt} presents a CI that uses the percentile of the bootstrap distribution
of the ELBO estimator.
This CI is based on the following approximation:
\begin{equation}
P\left(\hat{\theta}^*_{ELBO,\ell}- \hat{\theta}_{ELBO,\ell}<t|X_1,\cdots,X_n\right) \approx P\left(\hat{\theta}_{ELBO,\ell}- \theta_{ELBO,\ell}<t\right).
\label{eq::CI_bt2}
\end{equation}
Namely,
the CDF of the difference between ELBO estimator and the truth $\hat{\theta}_{ELBO,\ell}- \theta_{ELBO,\ell}$
can be approximated by the CDF of the bootstrapped differences.
Thus, $\hat{G}_\ell$ approximates the distribution of the actual difference
and we use it to construct a $(1-\alpha)$ CI.
We will show
the validity of equation \eqref{eq::CI_bt2}
in Theorem~\ref{thm::fixed} and Theorem~\ref{thm::increase}.

\begin{figure}[htb]
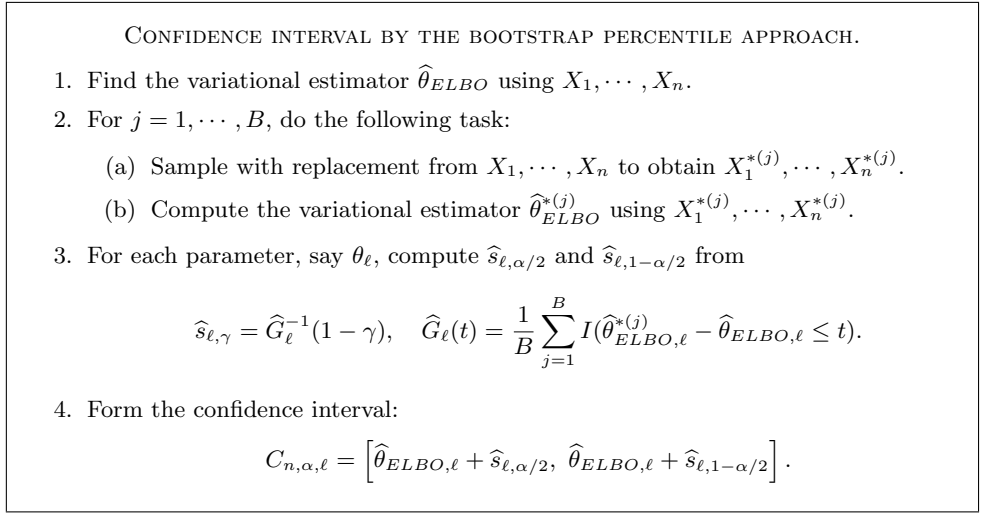

\fbox{\parbox{5in}{
\begin{center}
{\sc Confidence interval by the bootstrap percentile approach.}
\end{center}
\begin{center}
\begin{enumerate}
\item Find the variational estimator $\hat{\theta}_{ELBO}$ using $X_1,\cdots,X_n$.
\item For $j=1,\cdots, B$, do the following task:
	\begin{enumerate}
	\item Sample with replacement from $X_1,\cdots,X_n$ to obtain
	$X^{*(j)}_1,\cdots,X^{*(j)}_n$.
	\item Compute the variational estimator $\hat{\theta}^{*(j)}_{ELBO}$ using $X^{*(j)}_1,\cdots,X^{*(j)}_n$.
	\end{enumerate}
\item For each parameter, say $\theta_\ell$, compute $\hat{s}_{\ell,\alpha/2}$ and $\hat{s}_{\ell,1-\alpha/2}$
from
$$
\hat{s}_{\ell,\gamma} =  \hat{G}_\ell^{-1}(1-\gamma), \quad \hat{G}_\ell(t) = \frac{1}{B}\sum_{j=1}^B I(\hat{\theta}^{*(j)}_{ELBO, \ell}-\hat{\theta}_{ELBO,\ell}\leq t).
$$
\item Form the confidence interval:
$$
C_{n,\alpha,\ell} = \left[\hat{\theta}_{ELBO,\ell}+\hat{s}_{\ell,\alpha/2},\,\, \hat{\theta}_{ELBO,\ell}+\hat{s}_{\ell,1-\alpha/2}\right].
$$
\end{enumerate}
\end{center}
}}
\caption{Confidence interval by the bootstrap percentile approach.}
\label{fig::alg::PA}
\end{figure}

\subsection{Confidence interval: the pivotal approach}	\label{sec::pivotal}

The (studentized) pivotal approach \citep{wasserman2006all},
also called a percentile-t approach \citep{hall2013bootstrap},
is another popular method for constructing a CI
and may lead to a CI with a higher-order correctness \citep{hall2013bootstrap}.

The pivotal approach requires a consistent estimator of the variance of $\hat{\theta}_{ELBO,\ell}$.
Let $\hat{\sigma}^2_{ELBO,\ell}$ be such a consistent estimator.
Note that $\hat{\sigma}^2_{ELBO,\ell}$ can be constructed using a sandwich estimator
as is described in \cite{hall2011asymptotic} and \cite{westling2015establishing}.
Then the statistic
$$
T_n = \frac{\hat{\theta}_{ELBO,\ell}-\theta_{ELBO,\ell}}{\hat{\sigma}_{ELBO,\ell}}
$$
acts as a t-statistic and converges to a standard normal distribution (see, e.g., equation \eqref{eq::elbo}).
Therefore, $T_n$ is a pivotal quantity that has asymptotic normality and
the pivotal approach is based on bootstrapping $T_n$
to construct a CI.

Instead of using the percentile from a standard normal distribution,
we use the bootstrap percentile of $T_n$.
For the $j$-th bootstrap sample,
we not only compute the bootstrap parameter estimate $\hat{\theta}^{*(j)}_{ELBO,\ell}$
but also re-compute the corresponding variance estimator $\hat{\sigma}^{*2(j)}_{ELBO,\ell}$
to evaluate the bootstrap version of the pivotal statistics
$$
T^{*(j)}_n = \frac{\hat{\theta}^{*(j)}_{ELBO,\ell}-\hat{\theta}_{ELBO,\ell}}{\hat{\sigma}^{*(j)}_{ELBO,\ell}}, \quad j=1,\cdots,B.
$$
We then pick the value $\hat{t}_{\ell,1-\alpha/2}$ as the $(1-\alpha/2)$ upper quantile of the empirical distribution function
of $|T^{*(1)}_n|,\cdots,|T^{*(B)}_n|$, i.e.,
$$
\hat{t}_{\ell,1-\alpha/2} = \hat{F}^{-1}_\ell(1-\alpha/2),\quad \hat{F}_\ell(t) = \frac{1}{B}\sum_{j=1}^B I\left(|T^{*(j)}_n|\leq t\right).
$$
The $(1-\alpha)$ CI is
\begin{equation}
C^\dagger_{n,\alpha,\ell} =
\left[\hat{\theta}_{ELBO,\ell}-\hat{\sigma}_{ELBO,\ell}\cdot\hat{t}_{\ell,1-\alpha/2},\quad \hat{\theta}_{ELBO,\ell}+\hat{\sigma}_{ELBO,\ell}\cdot\hat{t}_{\ell,1-\alpha/2}\right].
\label{eq::pivot}
\end{equation}
Note that $\hat{\sigma}_{ELBO,\ell}$ is the estimator of the variance of $\hat{\theta}_{ELBO,\ell}$
using the original sample.
Figure~\ref{fig::alg::pivot} provides a summary of the bootstrap pivotal approach
for constructing a CI.

The intuition of the bootstrap studentized pivotal approach is that
the distribution of bootstrap statistic $T^*_n$ (given $X_1,\cdots,X_n$)
converges to the distribution of $T_n$ faster than
the convergence of $T_n$ to a standard normal distribution. Thus, the CI in equation \eqref{eq::pivot}
has a higher order correctness \citep{singh1981asymptotic,babu1983inference,hall2013bootstrap}.


\begin{figure}[htb]
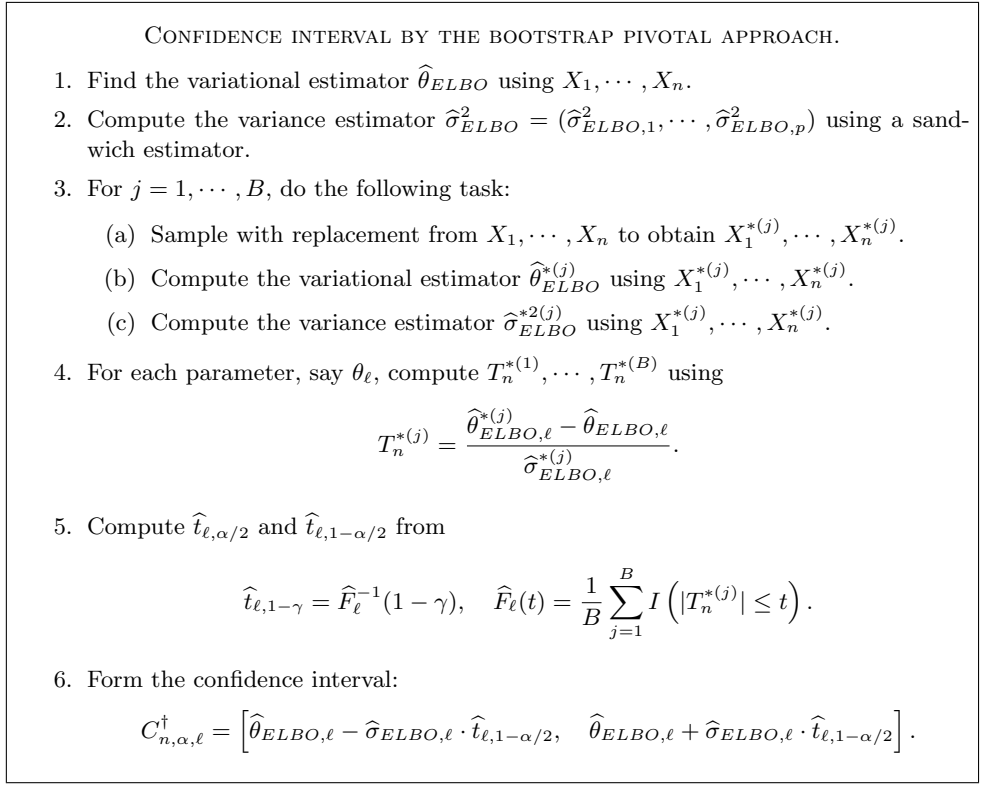

\fbox{\parbox{5in}{
\begin{center}
{\sc Confidence interval by the bootstrap pivotal approach.}
\end{center}
\begin{center}
\begin{enumerate}
\item Find the variational estimator $\hat{\theta}_{ELBO}$ using $X_1,\cdots,X_n$.
\item Compute the variance estimator $\hat{\sigma}^2_{ELBO} = (\hat{\sigma}^2_{ELBO,1}, \cdots, \hat{\sigma}^2_{ELBO,p})$
using a sandwich estimator.
\item For $j=1,\cdots, B$, do the following task:
	\begin{enumerate}
	\item Sample with replacement from $X_1,\cdots,X_n$ to obtain
	$X^{*(j)}_1,\cdots,X^{*(j)}_n$.
	\item Compute the variational estimator $\hat{\theta}^{*(j)}_{ELBO}$ using $X^{*(j)}_1,\cdots,X^{*(j)}_n$.
	\item Compute the variance estimator $\hat{\sigma}^{*2(j)}_{ELBO} $ using $X^{*(j)}_1,\cdots,X^{*(j)}_n$.
	\end{enumerate}
\item For each parameter, say $\theta_\ell$, compute $T^{*(1)}_n,\cdots,T^{*(B)}_n$
using
$$
T^{*(j)}_n = \frac{\hat{\theta}^{*(j)}_{ELBO,\ell}-\hat{\theta}_{ELBO,\ell}}{\hat{\sigma}^{*(j)}_{ELBO,\ell}}.
$$

\item Compute $\hat{t}_{\ell,\alpha/2}$ and $\hat{t}_{\ell,1-\alpha/2}$
from
$$
\hat{t}_{\ell,1-\gamma} = \hat{F}^{-1}_\ell(1-\gamma),\quad \hat{F}_\ell(t) = \frac{1}{B}\sum_{j=1}^B I\left(|T^{*(j)}_n|\leq t\right).
$$
\item Form the confidence interval:
$$
C^\dagger_{n,\alpha,\ell} =
\left[\hat{\theta}_{ELBO,\ell}-\hat{\sigma}_{ELBO,\ell}\cdot\hat{t}_{\ell,1-\alpha/2},\quad \hat{\theta}_{ELBO,\ell}+\hat{\sigma}_{ELBO,\ell}\cdot\hat{t}_{\ell,1-\alpha/2}\right].
$$
\end{enumerate}
\end{center}
}}
\caption{Confidence interval by the bootstrap pivotal approach.}
\label{fig::alg::pivot}
\end{figure}

\begin{remark}[Parametric bootstrap]		\label{rm::parametric}
In addition to the above bootstrap methods,
the parametric bootstrap is another popular approach which generates bootstrap samples from $P_{\hat{\theta}_{ELBO}}$ instead of the empirical distribution function.
However, we caution against using the parametric bootstrap.
When using the variational estimator,
the parametric bootstrap may not give a CI with the (asymptotic) nominal coverage even if the parametric family is correct (i.e., there exists $\theta_0\in\Theta$ such that
the data generating distribution $P = P_{\theta_0}$) because
the ELBO estimator $\hat{\theta}_{ELBO}$ does not converge to $\theta_0$ in general.
Thus, $P_{\hat{\theta}_{ELBO}}$ will not be close to $P_{\theta_0}$,
so there is no guarantee that the CI will have nominal coverage.
However, when the model is correctly specified and $\hat{\theta}_{ELBO}$ does converge to $\theta_0$
(this may occur when we allow $s$ to increase; see Remark~\ref{rm::s}),
the parametric bootstrap can provide CIs with nominal coverage;
see \cite{bickel2013asymptotic} for an example in the case of stochastic block model.
\end{remark}

\begin{remark}[Label Switching Problem]
In some models, the MLE may only be unique up to permutation of indices (see, e.g., the 
example in Section~\ref{sec::NLTCS}). 
In this case, 
the ELBO is non-convex so we need to use a gradient ascent method such as the EM algorithm. 
For each bootstrap sample, we will apply the EM algorithm
with the same initialization (we recommend to use
the estimator of the original sample as the initial point for each bootstrap sample). 
This will avoid the problem of 
label switching \citep{redner1984mixture}
and the bootstrap will recover the uncertainty in parameter estimation. 

\end{remark}

\section{Asymptotic Distribution and Bootstrap Consistency}	\label{sec::theory}

In this section, we derive the asymptotic distribution of the variational estimator
and its bootstrap theory.
We will
study the theory in both
scenarios: fixing and increasing $d$, the dimension of parameters $\theta$, after introducing further notation.

Let $B(x,r)$ be a ball with radius $r$ centered at $x$.
We define $\Psi(\theta) = \E ({\sf \cE}(\theta|X_1))$,
and let $\Psi_\theta = \nabla \Psi$ and $\Psi_{\theta \theta} = \nabla\nabla \Psi$ to be the gradient and Hessian matrix of $\Psi$, respectively.
For a unit vector $b\in \R^d$ and a function $f:\R^d\mapsto \R$,
$\nabla_b f = b^T \nabla f$ is the derivative of $f$ in the direction of $b$.
For a matrix $A\in\R^{d\times d}$, we denote $\lambda_{\max}(A)$ and $\lambda_{\min}(A)$ to be the largest and
the smallest eigenvalues of $A$,
respectively.

\subsection{Fixed Dimension}

When the dimension $d$ is fixed,
the ELBO estimator and its target can be analyzed using the theory of M-estimators \citep{van1998asymptotic}.
The asymptotic normality of $\hat{\theta}_{ELBO}-\theta_{ELBO}$ has been
analyzed in the literature \citep{hall2011asymptotic,bickel2013asymptotic,westling2015establishing,wang2017frequentist}
under several scenarios.
Here we present the asymptotic normality using the result stated in \cite{westling2015establishing}
because they also considered frequentist estimation in the general context of latent variable models.  
\begin{thm}[Theorem~2 in \cite{westling2015establishing}]
Assume conditions (B1-B5) in the appendix of \cite{westling2015establishing}.
Then
\begin{equation}
\sqrt{n}\left(\hat{\theta}_{ELBO}-\theta_{ELBO}\right) \overset{D}{\rightarrow} N(0, V(P_0,\theta_{ELBO})),
\label{eq::elboAsymp}
\end{equation}
where $V(P_0,\theta_{ELBO}) = A(P_0,\theta_{ELBO})^{-1}B(P_0,\theta_{ELBO})A(P_0,\theta_{ELBO})$ is a $p\times p$ matrix
such that
\begin{align*}
A(P_0,\theta_{ELBO}) &= \E_{X\sim P_0}(\Psi_{\theta\theta}(\theta_{ELBO}|X)),\\
B(P_0,\theta_{ELBO}) &= \E_{X\sim P_0}(\Psi_{\theta}(\theta_{ELBO}|X)\Psi_{\theta}(\theta_{ELBO}|X)^T).
\end{align*}
\label{thm::fixed_asymptotics}
\end{thm}
We include the assumptions (B1-5) in the appendix of \cite{westling2015establishing} in appendix \ref{sec::B15}.
These assumptions are made to derive the asymptotic normality of
an $M$-estimator (see, e.g., Theorem 5.23 of \citealt{van1998asymptotic}). 
Essentially, these assumptions assure that $\omega_{\max}(\theta|x),{\sf ELBO}(\theta,\omega|x)$, and $\cE(\theta|x)$ are
well-defined and sufficiently smooth
and well-behaved around $\theta_{ELBO}$ and $P_0$-a.e. $x$.
Viewing the ELBO estimator as the MLE, the quantity $\Psi_\theta(\cdot)$
and $V(P,\theta_{ELBO})$
are analogous to the score function and the Fisher information matrix, respectively.

To describe the validity of a bootstrap procedure,
we often use the notion of
convergence under Kolmogorov distance \citep{van1998asymptotic}.
For two random variables $A$ and $B$, their Kolmogorov distance is
$$
\sup_t \left|P\left(A<t\right)-P\left(B<t\right)\right|.
$$
The bound on Kolmogorov distance is also called the Berry-Esseen bound \citep{berry1941accuracy,esseen1942liapounoff}.
Note that convergence in probability in Kolmogorov distance is a stronger
result, compared to convergence in distribution. Namely,
if a sequence of random variables $A_1,\cdots,A_n,\cdots$ with $d_K(A_n,A_0)\overset{P}{\rightarrow}0$,
then $A_n\overset{D}{\rightarrow} A_0$.

Let $\Delta_n = \sqrt{n}\left(\hat{\theta}_{ELBO}-{\theta}_{ELBO}\right)$ and
$\Delta_n^* = \sqrt{n}\left(\hat{\theta}^*_{ELBO}-\hat{\theta}_{ELBO}\right)$ be the scaled difference
and the bootstrap version of it.
We will prove that $\Delta_n$ and $\Delta_n^*$ converge in Kolmogorov distance.

\begin{thm}
Assume conditions (B1-5) in the appendix of \cite{westling2015establishing}
and $\E\|\Psi_\theta(\theta_{ELBO}|X_1)\|^3<\infty$.
Then for any vector $a\in \R^{d}$ such that $\|a\|=1$,
\begin{equation}
\sup_{t} \left|P\left(a^T \Delta^*_n<t|X_1,\cdots,X_n\right)-P\left(a^T\Delta_n<t\right)\right|
\overset{P} {\rightarrow}0.
\end{equation}
Thus, for any $\ell=1,\cdots, d$,
\begin{align*}
P(\theta_\ell \in C_{n,\alpha,\ell})& {\rightarrow}  1-\alpha\\
P(\theta_\ell \in C^\dagger_{n,\alpha,\ell})&{\rightarrow}  1-\alpha,
\end{align*}
where $C_{n,\alpha,\ell}$ and $C^\dagger_{n,\alpha,\ell}$
are the CIs based on equations \eqref{eq::CI_bt} and \eqref{eq::pivot}, respectively.
\label{thm::fixed}
\end{thm}
The proof is deferred to the appendix.
Theorem~\ref{thm::fixed} shows that no matter which orientation we project onto (using the unit vector $a$),
the distribution of random variable $\hat{\theta}_{ELBO}-{\theta}_{ELBO}$
and the distribution of its bootstrap variant $\hat{\theta}^*_{ELBO}-\hat{{\theta}}_{ELBO}$
converge.
Thus,
the bootstrap quantile converges to the quantile of the actual distribution,
which proves validity of the bootstrap.


\subsection{Increasing Dimension}

We now study the bootstrap theory when the dimension of parameters is allowed to increase with respect to the sample size.
These situations occurs in many scenarios.
For example, in a mixed membership model,
we may want to increase the number of subgroups
when we have a larger sample.
Or in an item response theory model, both the number of questions in a test and the number of participants
may be increasing at the same time \citep{haberman1977maximum,douglas1997joint}.
In this case, we will write $d=d_n \rightarrow\infty$ as $n\rightarrow\infty$.
Note that we only allow $d$, the dimension of $\theta$ increase and the dimension of variational parameters
are assumed to be fixed. 
Thus, the population quantity $\theta_{ELBO}$ will also be changing.

{\bf Assumptions.}
\begin{itemize}
\item[(A0)] $\theta_{ELBO}\in\Theta$ is the unique maximizer of $\Psi(\theta)$ and $\omega_{\max}(\theta|x)$ is unique for each $\theta\in\Theta$
and almost surely for $x\in\R^J$ under $P_0$.
\item[(A1)] There exists $c_0>0$ such that all eigenvalues of $\Psi_{\theta\theta}(\theta_{ELBO})$ are not in $[-c_0,c_0]$
for any $d$.
\item[(A2)] There exists $r_0,c_1>0$ such that for any unit vectors $b_1,b_2,b_3\in \R^d$,
$$
\sup_{x}|\nabla_{b_1}\nabla_{b_2}\nabla_{b_3}{\sf \cE}(\theta|x)|\leq c_1<\infty
$$
for all $\theta \in B(\theta_{ELBO}, r_0)$ and $d$.
\item[(A3)] There exists $c_2>0$ such that for any unit vector $a\in\R^d$,
$$
\E\left(|\nabla_a{\sf \cE}(\theta|X_1)|^3\right)\leq c_2<\infty
$$
for any $d$.
\end{itemize}

(A0) is a very common assumption that requires $\theta_{ELBO}$ to be uniquely defined \citep{westling2015establishing}.
Note that we can relax (A0) to require $\theta_{ELBO}$ to be unique under permuting the indices
when the model is symmetric (such as the example in Section~\ref{sec::NLTCS}).
The theoretical results will be the same after a small modification to the proof so here we make this assumption to simplify the exposition.
(A1) implies that the Hessian matrix of $\Psi$ is invertible at $\theta_{ELBO}$ when $n,d\rightarrow \infty$.
This is a generalization of the invertible Fisher information matrix condition to the increasing-dimensional setting.
(A2) can be viewed as a generalization of a bounded $2$-norm of the third derivative tensor
$\nabla\nabla\nabla {\sf \cE}(\theta_{ELBO}|x)$.
To see this, consider only two-directional derivative, $|\nabla_{b_1}\nabla_{b_2} \cE(\theta_{ELBO}|x)|$.
The supremum of this will be the $2$-norm (maximum absolute eigenvalue) of $\cE (\theta_{ELBO}|x)$.
Note that assumptions similar to (A1--2) also appear in \cite{portnoy1985asymptotic} and \cite{mammen1989asymptotics}.
(A3) is a third moment condition that is used to establish a Berry-Esseen bound \citep{berry1941accuracy,esseen1942liapounoff}.
Note that when $d$ is changing with respect to $n$, (A2) and (A3) can be relaxed
so that constants $c_1$ and $c_2$ can depend on $n$. However, this 
relaxation will put another constraint on how fast $d\rightarrow n$ with respect to $n\rightarrow\infty$.

Note that we do not assume the distribution $P_\theta = P(x,z;\theta)$ belongs
to an exponential family.
If $P_\theta$ belongs the exponential family, the assumptions can be weakened
to the assumptions in \cite{portnoy1988asymptotic}.

\begin{thm}
Assume (A0-3) and $d = d_n\rightarrow \infty$ and $\frac{d^2}{n}\rightarrow0$.
Then, for any vector $a_n\in\R^{d}$ such that $\|a_n\|=1$, there exists a number $v(a_n)$ such that
\begin{equation}
\sup_{t} \left|P\left(a_n^T \Delta_n<t\right)-P\left(\sigma(a_n)\cdot Z<t\right)\right|
 {\rightarrow}0,
\end{equation}
where
$Z$ is a standard normal random variable.
Moreover,
\begin{equation}
\sup_{t} \left|P\left(a_n^T \Delta^*_n<t|\mathcal{X}\right)-P\left(a_n^T\Delta_n<t\right)\right|
\overset{P} {\rightarrow}0.
\end{equation}
Thus, for any $\ell=1,\cdots, d$,
\begin{align*}
P(\theta_\ell \in C_{n,\alpha,\ell})&{\rightarrow}  1-\alpha\\
P(\theta_\ell \in C^\dagger_{n,\alpha,\ell})&{\rightarrow}  1-\alpha,
\end{align*}
where $C_{n,\alpha,\ell}$ and $C^\dagger_{n,\alpha,\ell}$
are the CIs based on equations \eqref{eq::CI_bt} and \eqref{eq::pivot}, respectively.
\label{thm::increase}
\end{thm}
The proof is deferred to the appendix.
The first assertion in Theorem~\ref{thm::increase} states that the difference between the ELBO estimator
and its target converges to a normal distribution when we project the difference to any direction.
The quantity $\sigma(a)$ is the standard deviation of the difference of the estimator
that depends on the data-generating distribution $P_0$ and on the variational family that is being used.
Note that, when the dimension is fixed, $\sigma(a) = a^T V(P_0, \theta_{ELBO})a$.

The second assertion in Theorem~\ref{thm::increase}
shows that the limiting distributions of the scaled difference and
its bootstrap variant are asymptotically the same.
This implies that the CI constructed using the bootstrap
or variance estimated by the bootstrap is asymptotically valid.

Note that the requirement $\frac{d^2}{n}\rightarrow0$
is very common in increasing-dimensional problem; see, e.g., \cite{mammen1993bootstrap,portnoy1988asymptotic}.

\begin{remark}[Increasing both $d$ and $s$]
Theorem~\ref{thm::increase} can be applied to a case where both the dimension of parameter $d$
and the dimension of variational parameter $s$ are increasing. 
In this case, we need assumptions (A0-A3) to hold for every $s$ and $d$. 
When we allow $s=s_n$ to increase, the assumption (A1) may be too strong.
We can relax this assumption by allowing the constant $c_0$ in (A1) to decrease to $0$ slowly. 
The increasing rate of $s_n$ will be constrained by the decreasing rate of $c_0$
to guarantee the invertibility of $\Psi_{\theta\theta}$.
\end{remark}

\begin{remark}[Increasing $s$ only]
\label{rm::s}
Even when $d$, the dimension of the parameter, remains fixed (i.e., the population MLE $\theta_{MLE}$ is fixed),
changing $s$, the dimension of $\omega$,
will also change the (population) quantity $\theta_{ELBO}=\theta_{ELBO,s}$.
In some situations, we even have
$\theta_{ELBO} = \theta_{ELBO,s}\rightarrow \theta_{MLE}$;
see \cite{hall2011asymptotic}
and \cite{bickel2013asymptotic} for examples.
The difference $\theta_{ELBO,s}- \theta_{MLE}$
can be viewed as the bias of the variational estimator.
Because the dimension of variational parameter $s$
can be viewed as a measure of model complexity
of the variational estimator,
the property $\theta_{ELBO,s}\rightarrow \theta_{MLE}$ can be interpreted as
an asymptotic unbiasedness property in terms of model complexity.
\end{remark}

\begin{remark}[High-dimensional case]
When $d>n$, the conventional central limiting theorem fails
because of the complexity coming from the high dimensional parameters \citep{portnoy1984asymptotic,portnoy1985asymptotic}.
Thus, CIs from the percentile or pivotal approaches
do not have the nominal coverage.
However, it is still possible to construct an asymptotically valid CI
using the bootstrap.
The rectangle CI \citep{chernozhukov2013gaussian}
is one example.
We refer the readers to \cite{chernozhukov2013gaussian,wasserman2013estimating,fan2015guarding}
for more details about rectangle CIs.
\end{remark}

\section{Data Analysis}	\label{sec::NLTCS}
We illustrate our theoretical results with multivariate binary data on functional disability from the National Long Term Care Survey (NLTCS).  \cite{erosheva2007describing} presented the first case of variational estimation for mixed membership models with binary data from the NLTCS. Here, we consider observations collected on the NLTCS participants in 1984, 1989 and 1994. The data contain binary indicators on 6 activities of daily living (ADL) and 10 instrumental activities of daily living (IADL) for community-dwelling elderly. 
The 6 ADL items include basic activities of hygiene and personal care: eating, getting in/out of bed, getting around inside, dressing, bathing, and getting to the bathroom or using toilet. The 10 IADL items include basic activities necessary to reside in the community: doing heavy housework, doing light housework, doing laundry, cooking, grocery shopping, getting about outside, traveling, managing money, taking medicine, and telephoning. Responses are coded as 0 and 1, where 1 denotes a presence and 0 denotes an absence of a functional disability. In the NLTCS, positive (1) ADL responses mean that during the past week the activity had not been, or was not expected to be, performed without the aid of another person or the use of equipment; negative (0) IADL responses mean that a person usually could not, or was not going to be able to, perform the activity because of a disability or a health problem. For a more in-depth discussion, see \citet{manton1993estimates}, and \citet{erosheva2006operational}.

Similar to \citet{erosheva2007describing}, we also use a mixed membership analysis.  \citet{erosheva2007describing} take a fully Bayesian approach and specify priors for the $\alpha$ and $\Pi$ model parameters discussed below; however, in this analysis we take a frequentist approach and directly compute maximum ELBO estimates for $\alpha$ and $\Pi$. Also, \citet{erosheva2007describing} analyze all four waves (1982, 1984, 1989, and 1994), but we restrict our analysis to the 1984, 1989, and 1994 waves. 

In particular, we are interested in two tasks. First, we use a mixed membership model to describe the 5,934 observations in the 1984 wave. We use a variational procedure to estimate the model parameters and then give bootstrapped confidence intervals for each of those estimates. Next, we consider the 4,463 and 5,089 observations from 1989 and 1994 respectively. We test whether the responses observed in 1989 and 1994 arise from the same distribution. 
Given the two natural sub-populations, this corresponds to a possible two-sample test described in Section \ref{sec:furtherConsid}. A conceptually simpler approach could be used instead of a model based approach. At the coarsest resolution, this might be a two sample t-test for each of the 16 variables, and at the finest resolution, this might be a two sample t-test for each of the $2^{16}$ possible response patterns. However, testing in the mixed membership framework allows investigation of subtle changes in the underlying structure, while still retaining easy interpretation.

\subsection{Mixed membership models and variational inference}
Throughout this analysis, we use mixed membership models to uncover latent structure. Like a mixture model, mixed membership models assume that the population is comprised of several groups, where each group has a distribution over the observed variables. However, while mixture models assume that each individual belongs to a single group, mixed membership models allow each individual to have a partial membership in multiple groups \citep{airoldi2014handbook}. Mixed membership models have been used for topic modeling \citep{blei2003latent}, social network analysis \citep{airoldi2008mixed}, survey data \citep{erosheva2007describing}, and statistical genetics \citep{pritchard2000inference}. Note that allowing for mixed membership differs from estimating the posterior probability of group assignment when using a mixture model. Under a mixture model, as the data about an individual grows, the posterior should concentrate on a single group, while in a mixed membership model, as the data about an individual grows, we may consistently estimate the individual's membership, which could be in the interior of the simplex. 

In the setting we consider, for each individual $i = 1, \ldots, n$ we observe multivariate data $X_i = (X_{i,1}, \ldots, X_{i,16})$ and assume the following generative model. Let $j = 1, \ldots J= 16$ index variables and $K$ be the fixed number of groups. We assume fixed parameters $\alpha \in \mathbb{R}^K_{> 0}$, which regulates the Dirichlet distribution for group membership, and $\Pi = \{\pi_{jk}\}$ for $j = 1, \ldots, J$ and $k = 1, \ldots, K$, where $\pi_{jk}$ is the Bernoulli parameter for a response to variable $j$ from a full member of group $k$. The generative model for individual $i$ is: 
\begin{enumerate}
  	\item $\lambda_i \sim \text{Dirichlet}(\alpha)$, where $\lambda_i$ lies in the $K-1$ simplex (i.e., $\sum_k \lambda_{ik} = 1$ and $\lambda_{ik}\geq 0$). Each element $\lambda_{ik}$ characterizes the extent of membership for individual $i$ in group $k$. 
  \item For each variable $j$:
  \begin{itemize}
  \item $g_{ij} \sim \text{Categorical}(\lambda_i)$, where $g_{ij} \in \{1, \ldots, K\}$ indicates the group whose parameters govern individual $i$'s response to question $j$.
  \item $X_{ij} \sim \text{Bernoulli}(\pi_{jg_{ij}})$, the observed response for individual $i$ on question $j$.
\end{itemize}
\end{enumerate}
This hierarchical model assumes that each individual responds to each question as a full member of group $g_{ij}$. However, for each individual, the group may vary across variables and the probability of responding as a full member of group $k$ for each question is governed by $\lambda_{ik}$. In addition, $X_{ij}$ is independent of $X_{ij'}$ given $\lambda_{ik}$. 

The parameters of interest are $\alpha$ and $\Pi$. For the Dirichlet parameter $\alpha$, the quantity $\alpha_{k} /\sum_{k'}\alpha_{k'}$ indicates the relative proportion of each group and the magnitude, $\sum_k \alpha_k$, indicates the level of intra-individual mixing. Distributions with larger values of $\sum_k \alpha_k$ concentrate density in the interior of the simplex and imply a higher level of intra-individual mixing, while distributions with smaller values of $\sum_k \alpha_k$ concentrate density in the corners of the simplex and indicate less intra-individual mixing. The Bernoulli parameters $\Pi$ characterize the ability/disability of each group. The parameters $\lambda_i$ and $g_{ij}$ are latent variables which we consider as nuisance parameters. In the previous notation, $\theta = \{\alpha, \Pi\}$ and $Z_{i} = \{\lambda_i, g_{ij}\}$.

Although the model is straightforward to describe and generate, maximum likelihood estimation is difficult because the normalizing constant is intractable. Thus, to fit the model, we use the \texttt{mixedMem} R package \citep{wang2015fitting} which specifies the following mean field variational distribution with variational parameters $\phi_i \in \mathbb{R}^K_{> 0}$ and $\delta_{ij}$ in the $K-1$ simplex:
\begin{equation}
  \begin{aligned}
  \lambda_i &\sim \text{Dirichlet}(\phi_i)\\
  g_{ij} &\sim \text{Categorical}(\delta_{ij});\\
  X_{ij} &\sim \text{Bernoulli}(\pi_{jg_{ij}}).
\end{aligned}
\end{equation}
In the previous notation, $\omega_i = \{\phi_i, \delta_{ij}\}$. The likelihood and specified variational distribution yield the following ELBO:

\begin{equation}\scriptsize
\begin{aligned}
{\sf ELBO}(\theta, \omega |X) =
& \sum_i \log\Gamma\left(\sum_k \alpha_k\right)  - \sum_{i,k}\log\Gamma\left(\alpha_k\right) + \sum_{i,k}(\alpha_k-1)\left[\Psi(\phi_{ik})- \Psi\left(\sum_k \phi_{ik}\right)\right] \\
& \quad + \sum_{i,j,k} \delta_{ijk}  \left[\Psi(\phi_{ik})- \Psi\left(\sum_k \phi_{ik}\right)\right]\\
&\quad + \sum_{i,j,k}\delta_{ijk} X_{ij}\log(\pi_{jk}) + \sum_{i,j,k} \delta_{ijk} (1-X_{i,j}) \log(1-\pi_{jkv})\\
&\quad -
 \sum_i \log\Gamma\left(\sum_k \phi_{ik}\right)  + \sum_{i,k}\log\Gamma(\phi_{ik})\\
 &\quad - \sum_{i,k}(\phi_{ik}-1)\left[\Psi(\phi_{ik})- \Psi\left(\sum_k \phi_{ik}\right)\right] - \sum_{i,j,k}\delta_{ijk}\log(\delta_{ijk})
\end{aligned}
\end{equation}
where $\Gamma(\cdot)$ is the gamma function and $\Psi(\cdot)$ is the digamma function which is the derivative of the log-$\Gamma$ function.  We maximize the ELBO with respect to the parameters of interest, $\alpha$ and $\Pi$, and the variational parameters, $\phi_i$ and $\delta_{ij}$, through a block coordinate ascent procedure which alternates between two steps. In the first step, holding $\alpha$ and $\Pi$ fixed, we compute the optimal variational parameters by iterative coordinate ascent. Then, holding the variational parameters fixed, we update $\alpha$ and $\Pi$ through a Newton-Raphson procedure. Because there is no closed-form solution for $\hat \delta_{ij}(\alpha, \Pi)$ and $\hat \phi_i(\alpha, \Pi)$, we can not easily compute a Hessian required for the sandwich estimator of \citet{westling2015establishing} or the pivotal confidence intervals summarized by Figure~\ref{fig::alg::pivot}. However, percentile based bootstrap confidence intervals and bootstrap variance estimates can be used.

\subsection{Initial analysis and bootstrapped standard errors}
We first select an appropriate number of groups, $K$, using a pseudo-BIC criterion:
\begin{equation}
  \text{pBIC} = p\log(n) - 2 \times {\sf ELBO}(\hat \theta_{ELBO}, \hat \omega_{ELBO} |X),
\end{equation}
where $p = K + J\times K$, the count of parameters $\alpha$ and $\Pi$. Because the ELBO is generally multi-modal, we use 1,000 random initialization points (for $\alpha$ and $\Pi$) for $K = 2, \ldots, 9$. For each $K$, we then select the resulting stationary point with the largest ELBO and compute the pseudo-BIC. Using many random restarts is important, because, as is typically the case, the ELBO defined by the mixed membership model and variational family we use is multi-modal. We see from the left panel of Figure~\ref{fig:BIC} that the ELBO of each stationary point can vary widely within each value of $K$. In the right panel of Figure~\ref{fig:BIC}, we plot only the lowest pBIC for each $K$; we see that the pBIC criteria leads us to select a 4 group model, though a 6 group model might also be appropriate. 

The estimated Bernoulli and Dirchlet parameters for the optimal 4 group model are presented in Figure~\ref{fig::modelParams} and \ref{fig::modelAlpha}. The confidence intervals shown in black are calculated by using the non-parametric percentile bootstrap summarized in Figure~\ref{fig::alg::PA} and the confidence intervals shown in red are calculated using the parametric percentile bootstrap.
The intervals used are post model-selection \citep{leeb2005model}.

\begin{figure}[htb]
\centering
\includegraphics[scale = .5]{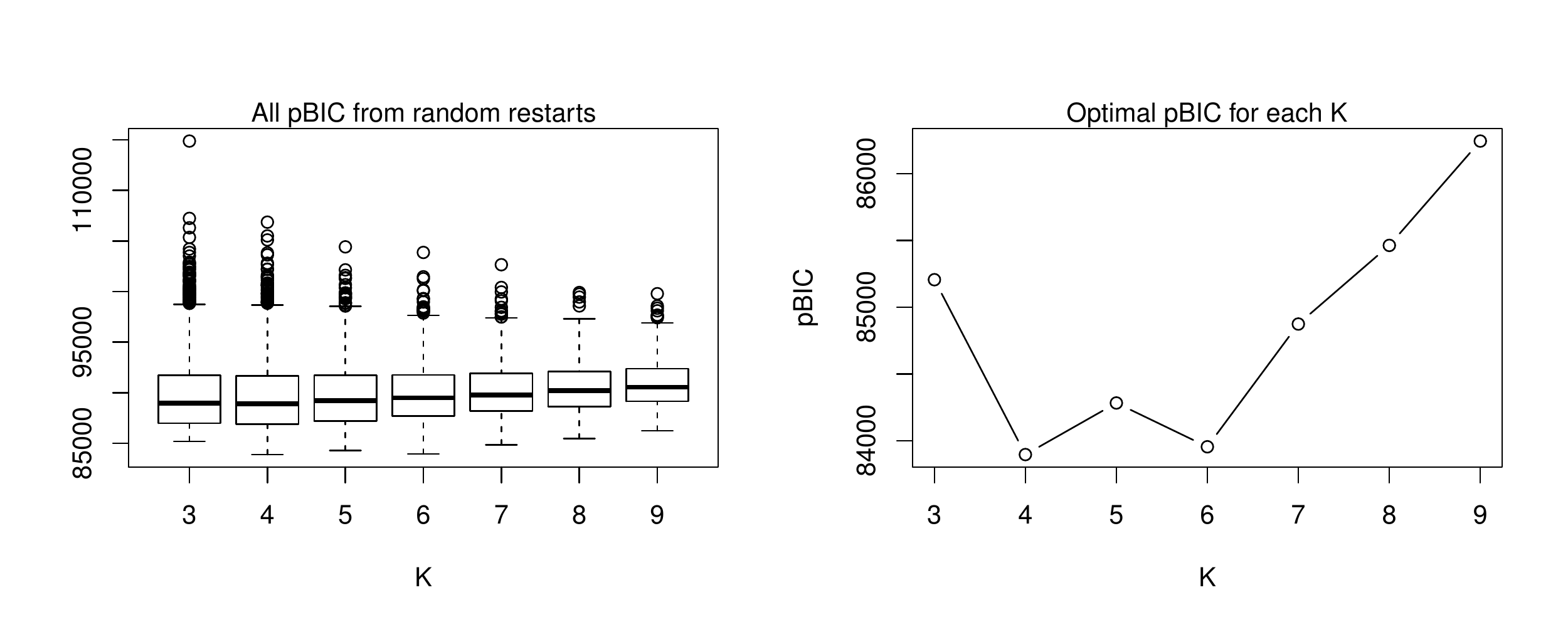}
\caption{\label{fig:BIC} The pseudo-BIC across levels of $K$ groups. The left panel shows pseudo-BIC values across all random initializations, and the right panel shows only the optimal model for each $K$.}
\end{figure}

Since the ELBO can be multi-modal and we use a coordinate ascent procedure, we need to carefully initialize each bootstrap run so that we do not enter another basin of attraction and overestimate the sampling variability. 
In particular, we initialize the global parameters, $\alpha$ and $\Pi$, 
as well as the individual latent variables,
$\lambda_i$ and $g_{ij}$, at the corresponding quantities estimated from the original data. 
In general, we expect each bootstrap run to require less computational effort than the original estimation procedure since we expect the initialization to be near the stationary point. In addition, each of the bootstrap runs can be easily parallelized on a cluster; for this particular analysis, an individual bootstrap run took roughly 15 seconds on a laptop.

In Figure~\ref{fig::modelParams}, we have sorted the groups top to bottom (1 through 4) by least disabled to most disabled. Group 1 is generally most likely to be able to perform each of the 16 tasks. Group 2 appears to be relatively less able to perform most physical/mobility related tasks, but is relatively more able to perform tasks requiring mental acuity. For instance, members of Group 2 are relatively less able to get in/out of bed, move around inside, and move around outside; however, they are relatively more able to cook, manage money, take medicine, and use the telephone. Group 3 appears to have more mobility, but is less able to perform tasks which require mental acuity. For instance, individuals in Group 3 are relatively more able to get in/out of bed, move around inside, and use the toilet, but less able to manage money or use the telephone. Finally Group 4 is generally least likely to be able to perform each task. The estimated Bernoulli parameters for Group 4 are higher than the marginal probabilities for all 16 tasks. 
Note that the CIs from the two bootstrap methods are small, indicating that our estimators are quite precise. 

\begin{figure}[htb]
\centering
\includegraphics[scale=.5]{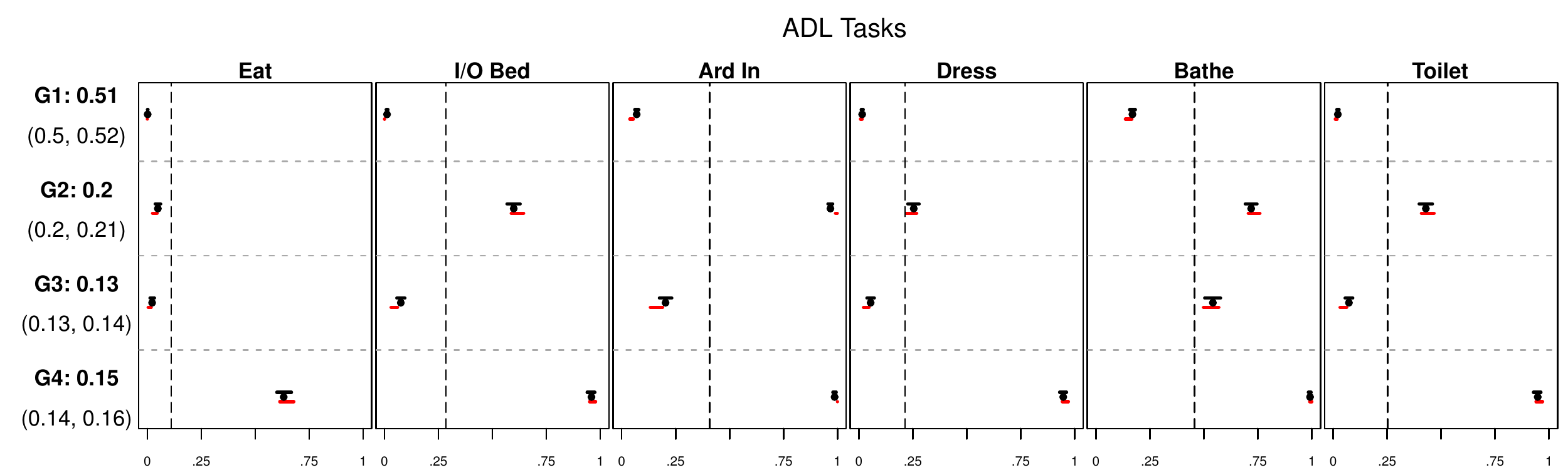}
\includegraphics[scale=.5]{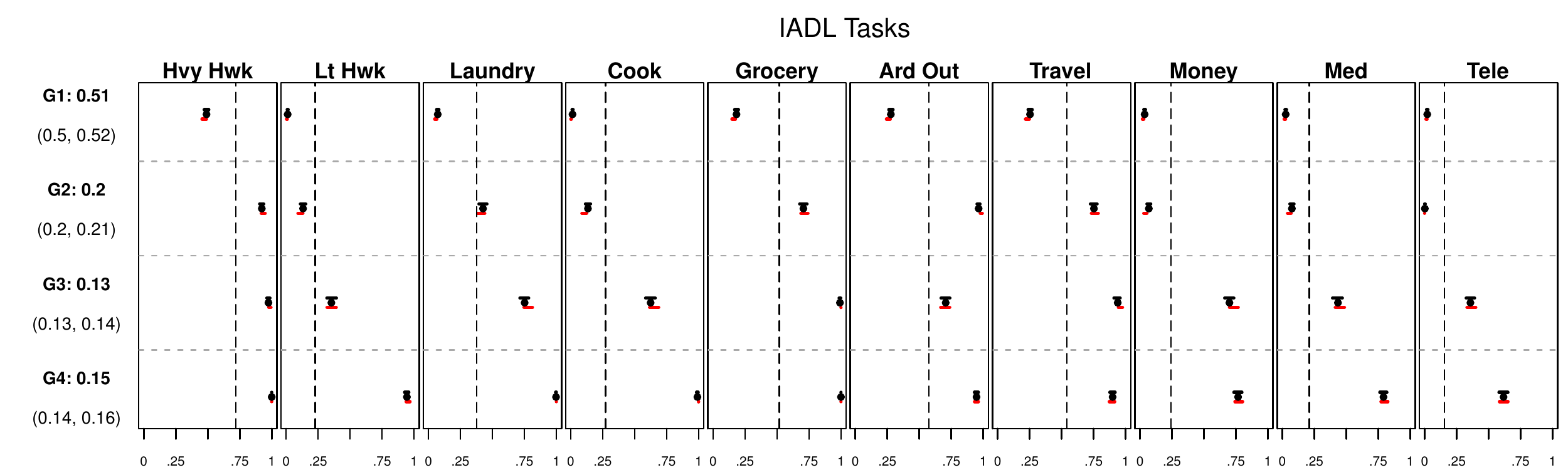}
\caption{\label{fig::modelParams}The estimated parameters and CIs for the 4 group mixed membership model. The black CI's are formed using the non-parametric bootstrap; the red CI's are formed using the parametric bootstrap. The top panel shows estimates for the ADL activities and the bottom panel shows estimates for the IADL activities. The estimated population proportion, $\hat \alpha_k/\sum_{k'} \hat \alpha_{k'}$ is shown on the left with the corresponding CI under each group label. For aiding interpretation, the vertical dashed lines shows the marginal proportion of individuals whose response was $1$ for each variable.}
\end{figure}

\begin{figure}[htb]
\centering
\includegraphics[scale=.5]{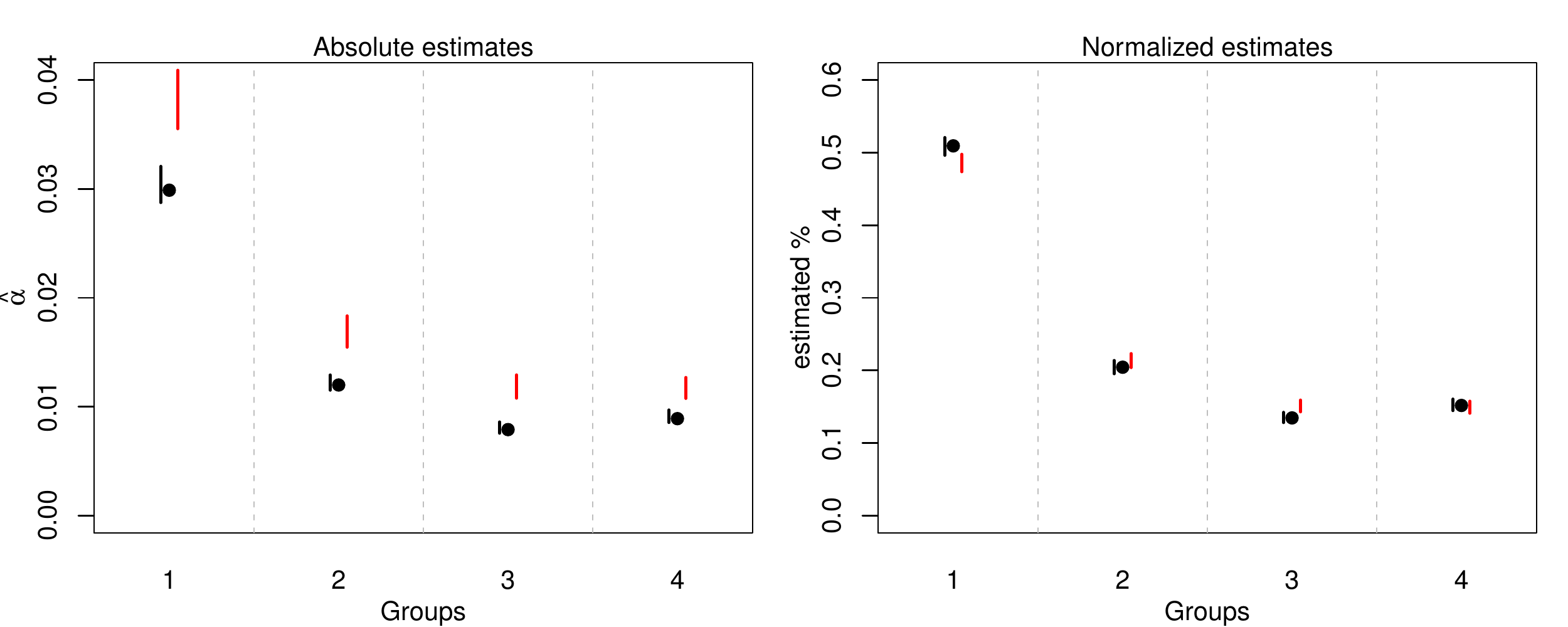}
\caption{\label{fig::modelAlpha}The left panel shows the estimates and CI's for $\alpha$ and the right panel shows the estimates and CI's for the proportion of each group; i.e., $\alpha_k / \sum_{k'} \alpha_{k'}$. The black CI's are formed using the non-parametric bootstrap and the red CI's are formed using the parametric bootstrap.}
\end{figure}

We caution against using the parametric bootstrap with variational inference since $\theta_{ELB0}$ in general is not equal to $\theta_{MLE}$ so the CI's may not always cover the variational point estimates. In particular, for the Bernoulli parameters, 37 of the 64 CI's constructed via the parametric bootstrap do not cover the point estimates. In addition, all 4 of the parametric bootstrap CI's for $\hat \alpha$ (and 2/4 of the CI's for the population proportions) do not cover the point estimates. However, all of the CI's (both for the Bernoulli and Dirichlet parameters) constructed using the non-parametric bootstrap do cover the point estimates. As noted by \citet{andrews2000inconsistency}, when the point estimates are near the boundary of the parameter space, the bootstrap estimates might be unstable. In this case, we see that when the Bernoulli parameters are close to 0 or 1, this generally causes a problem for the parametric bootstrap, but not for the non-parametric bootstrap.

\subsection{Two-sample test}	\label{sec::NLTCS::two}
We now consider observations from the 1989 and 1994 waves. In particular, we test whether the functional disability measures taken five years apart are generated by the same distribution. In order to concretely interpret differences between the two waves, we fix $K = 4$ and use the Bernoulli parameters estimated from the 1984 wave. We then find point estimates $\hat \alpha_{1989}$ and $\hat \alpha_{1994}$ separately by maximizing the ELBO with respect to $\alpha$ (keeping $\Pi$ fixed). Again, because of multi-modailty of the ELBO, we use 1000 random initialization to select an $\hat \alpha$ for each wave. In principal, fixing the Bernoulli parameters to any random quantity and concluding that $\alpha_{89, \text{ELBO}|\Pi} \neq \alpha_{94, \text{ELBO}|\Pi} $ would result in rejecting the null hypothesis (where $\alpha_{\text{ELBO}|\Pi}$ indicates the $\alpha$ value which maximizes the ELBO for fixed $\Pi$). However, we use the point estimates from the 1984 wave to facilitate interpretability.  

The estimated group proportions for 1989 and 1994 are shown in Figure~\ref{fig::alphaSplit} with the corresponding confidence intervals formed by the non-parametric percentile bootstrap standard errors. In 1994, the prevalence of the least disabled group (Group 1) increased, while the prevalence of Group 2 (incapable of mobility tasks, but capable of mental tasks), Group 3 (capable of mobility tasks, but incapable of mental tasks) and Group 4 (generally incapable of all tasks) all decreased by roughly .03 each.   

\begin{figure}[ht]
\centering
\includegraphics[scale=.55]{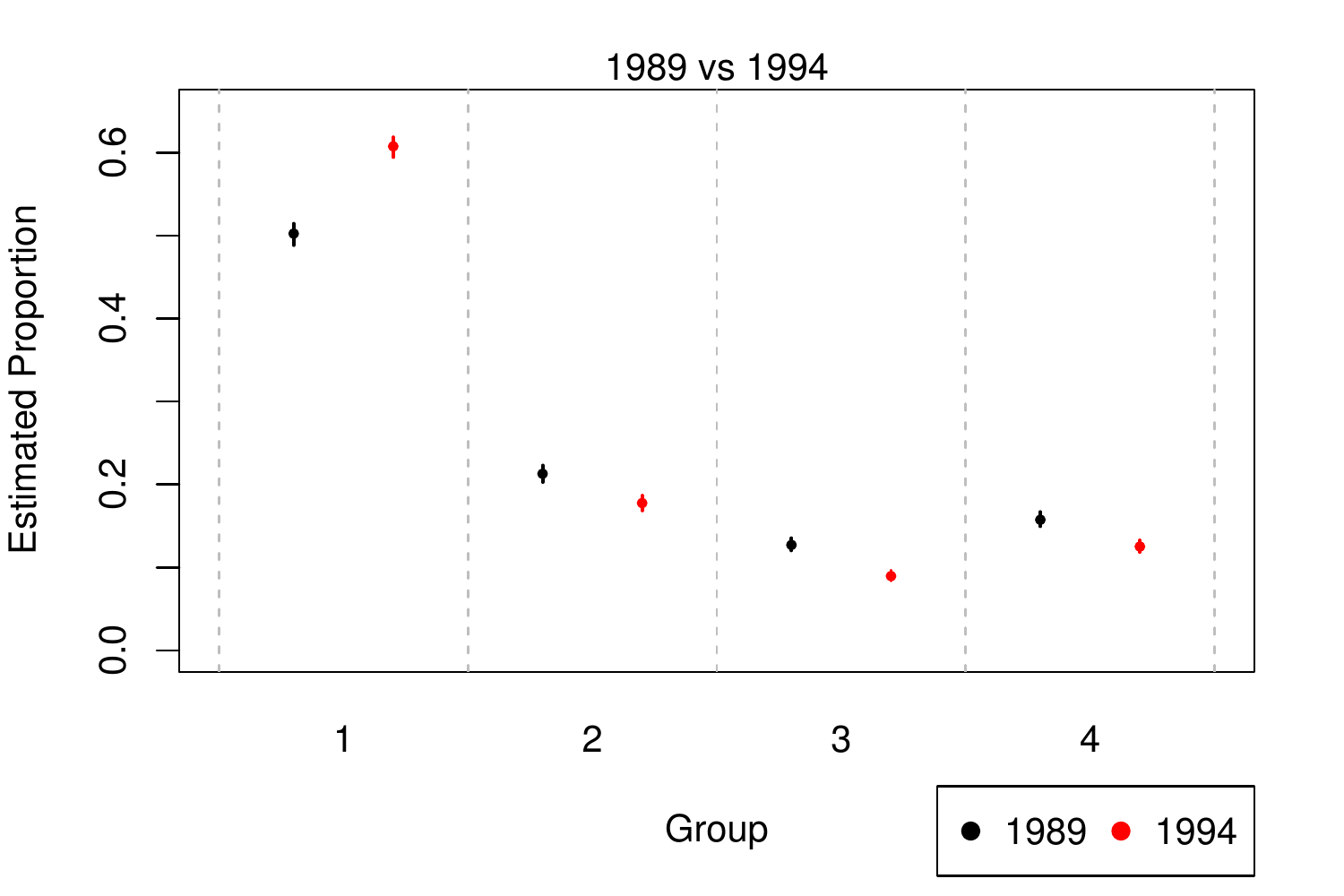}
\caption{\label{fig::alphaSplit}The estimated proportions of each group for 1989 and 1994. The CIs are percentile method bootstrapped intervals.}
\end{figure}

For good measure, we also perform a Wald test for the population proportions:
\[\hat p = (\hat \alpha_1, \hat \alpha_2, \hat \alpha_3) / \sum_{k=1}^4 \hat \alpha_k\]
where the proportion for Group 4 is excluded so that the distribution is non-degenerate. Using the bootstrap estimate of covariance $\widehat {V(\hat p)}$ calculated by the procedure summarized in Figure~\ref{fig::alg::var}, we find that
\[ (\hat p_{1989} - \hat p_{1994})^T\left(\widehat {V(\hat p_{1989})} + \widehat {V(\hat p_{1994})} \right)^{-1}(\hat p_{1989} - \hat p_{1994}) = 143.7\]
Thus, we reject the null hypothesis that all population proportions are equal with a p-value less than $10^{-16}$ when compared to the $\chi^2_3$ reference distribution.

\section{Discussion}	\label{sec::discussion}

%
%
%

We conclude this paper by including some remarks about practical aspects of the variational inference and the bootstrap, and by making several observations on the connections between the research presented in this paper and the work of late Stephen E. Fienberg, to whom this special issue of the {\it Annals of Applied Statistics} is dedicated. 

\paragraph{Bootstrap versus asymptotic normality}
For practitioners, constructing a CI using
a bootstrap approach is generally easier
than using the asymptotic normality \citep{hall2011asymptotic,westling2015establishing}.
To construct a CI using asymptotic normality, we need a (consistent) variance estimator and calculating
that estimator
often requires an involved derivation, which could be very challenging
when the model is complex. The NLTCS example of mixed membership is one such example.
And sometimes, such an estimator does not exist so we are unable
to use the asymptotic normality approach.
On the other hand, the implementation of a bootstrap approach is very easy --
it is just sampling with replacement and re-applying the variational inference.
The bootstrap approach does not require a consistent variation estimator
so it is a more general approach than the interval from asymptotic normality approach.
Moreover, if we do have a variance estimator,
as is discussed in Section~\ref{sec::pivotal},
we can construct a CI using the bootstrap pivotal approach,
which may lead to a CI with a higher order correctness \citep{singh1981asymptotic,babu1983inference,hall2013bootstrap}.

\paragraph{Implications for variational inference in Bayesian settings}
Theorems~\ref{thm::fixed_asymptotics} and \ref{thm::fixed}
proved the asymptotic normality and bootstrap validity
of using the variational inference to approximate the MLE.
These theorems
can be applied
to Bayesian variational estimators
as long as the prior is sufficiently smooth
(for a trivial case, consider a uniform prior on the parameter space)
or to a penalized ELBO with a very week (i.e., asymptotically negligible) penalty.
However, in the Bayesian framework, the posterior distribution and credible intervals are
the quantities of the interest and the CI is not the main objective.
For the penalized ELBO, a weak penalty is often not of research interest
because it neither encourages sparsity nor stabilizes the estimator.

\paragraph{Two-sample test and comparison}
The two-sample test used in Section~\ref{sec::NLTCS::two}
shows great potential of combining the bootstrap CI
and variational inference.
In the NLTCS example that we presented, without adequate tools to obtain uncertainty in estimates, it is possible that an erroneous conclusion could have been made, stating that the proportion of responses that corresponds to profile 3 (mainly problems with managing money, grocery shopping, and traveling) stays the same over the ten year interval, while our analysis demonstrated that the difference is significant. 
Note that the approach of comparing two samples
is very generic -- it can be applied to
various problems involving a comparison of two datasets using variational inference.
With the methodologies developed in this paper, we can
assess the significance of the difference between estimates using the bootstrap
and make a statistical conclusion about the two datasets.




\paragraph{Connection with Stephen E. Fienberg}
Stephen Fienberg had originally introduced Erosheva, then a graduate student, to the Grade of Membership Model~\citep{woodbury1978mathematical} and to the functional disability data from the NLTCS.  Erosheva's graduate work has motivated the original NLTCS publication \citep{erosheva2007describing} as well as the development of the general mixed membership modeling framework.
For that original NLTCS publication, under Fienberg's direction, Erosheva and Joutard have developed and implemented both a fully Bayesian MCMC approach and the corresponding variational estimation algorithm for mixed membership models with binary data \citep{erosheva2007describing}. Later, \citet{wang2015variational}, extended this line of work and developed a variational estimation algorithm for mixed membership models with rank data, where, at a suggestion of a reviewer, they used bootstrap methods to assess uncertainty in the model estimates.

Although Fienberg's impact on statistical science spanned many areas, including the census and survey research in general, he is perhaps best known for his contributions to discrete data analysis and log-linear models. Even though he used to say that ``everything is a log-linear model", meaning that almost any statistical model for discrete data has a log-linear representation, he did not brush off other approaches. In particularly, Stephen Fienberg was a big advocate of mixed membership models -- because of their flexibility and practical appeal -- during the last 15 years of his career. Mixed membership models present certain challenges in estimation, and, while recommending variational inference as a step toward solving those challenges, Fienberg was very much cognizant of both the practical advantages and the lack of statistical theory for variational estimation. We are not aware of his involvement in recent efforts to provide a theoretical foundation for variational estimates, but we can say with confidence that he would have been supportive and encouraging for advancing research in this direction.

\section*{Acknowledgements}

We thank the referee and the editor for insightful comments.
We also thank Fang Han for helpful comments on the theoretical results.

\appendix

\section{Proofs}

\begin{proof}[Proof of Theorem~\ref{thm::fixed}]

Our proof consists of two parts.
In the first part, we show that the asymptotic normality admits a Berry-Essen bound.
In the second part, we show that the bootstrap variant
converges to the same distribution with a Berry-Essen bound.

{\bf Part 1: Berry-Esseen Bound.}
By the derivation of Theorem~2 in \cite{westling2015establishing} and Theorem 5.23 in \cite{van1998asymptotic},
the ELBO estimator has the property that
$$
\hat{\theta}_{ELBO} - \theta_{ELBO} = A(P_0,\theta_{ELBO}) \cdot \frac{1}{n}\sum_{i=1}^n \Psi_\theta(\theta_{ELBO}|X_i) +o_P\left(\frac{1}{\sqrt{n}}\right).
$$
Note that the above equation is a common expression for an $M$-estimator.
Thus, $a^T\Delta_n$ has the following expression:
\begin{align*}
a^T\Delta_n  &= a^T \sqrt{n}(\hat{\theta}_{ELBO} - \theta_{ELBO})\\
&= a^TA(P_0,\theta_{ELBO}) \cdot \frac{1}{\sqrt{n}}\sum_{i=1}^n \Psi_\theta(\theta_{ELBO}|X_i)  + o_P(1)\\
&= \sqrt{n}\cdot\frac{1}{n}\sum_{i=1}^n (W_i-\E(W_i)) + o_P(1),
\end{align*}
where $W_i = a^TA(P_0,\theta_{ELBO})\Psi_\theta(\theta_{ELBO}|X_i)$
and $\E(W_i)=0$
because $\E\left(\Psi_\theta(\theta_{ELBO}|X_1)\right)=0$.

By the assumption that $\E\|\Psi_\theta(\theta_{ELBO}|X_1)\|^3<\infty$ and the Berry-Esseen theorem \citep{berry1941accuracy,esseen1942liapounoff},
we conclude that
$$
\sup_{t}|P(a^T\Delta_n<t) - P(Z_w<t)| \leq C_{BE} \frac{\E\|W_1\|^3}{\sqrt{n}} + o_P(1),
$$
where $Z_w$ is a normal distribution with variance ${\sf Var}(W_1)$
and $C_{BE}$ is a universal constant from the Berry-Esseen theorem.

{\bf Part 2: bootstrap.}
Let $\mathcal{X}_n = \{X_1,\cdots,X_n\}$.
For the bootstrap case, we have a similar decomposition of $\hat{\theta}^*_{ELBO} - \hat{\theta}_{ELBO}$:
$$
\hat{\theta}^*_{ELBO} - \hat{\theta}_{ELBO} = A(\hat{P}_n,\hat{\theta}_{ELBO}) \cdot \frac{1}{n}\sum_{i=1}^n \Psi_\theta(\hat{\theta}_{ELBO}|X^*_i) +E^*_n,
$$
where $\|E^*_n \| = o_P\left(\frac{1}{\sqrt{n}}\right)$ is a small correction error
and $\hat{P}_n$ is the empirical distribution.
Thus,
\begin{align*}
a^T\Delta^*_n  &= a^T \sqrt{n}(\hat{\theta}^*_{ELBO} - \hat{\theta}_{ELBO})\\
&= \sqrt{n}\cdot\frac{1}{n}\sum_{i=1}^n (W^*_i-\E(W^*_i|\mathcal{X}_n)) + o_P(1),
\end{align*}
where $W^*_i = a^TA(\hat{P}_n,\hat{\theta}_{ELBO})\Psi_\theta(\hat{\theta}_{ELBO}|X^*_i)$
and $\E(W^*_i|\mathcal{X}_n)= 0 $.
Again, by applying the Berry-Essen theorem \citep{berry1941accuracy,esseen1942liapounoff},
we conclude that
$$
\sup_{t}|P(a^T\Delta^*_n<t|\mathcal{X}_n) - P(Z^*_w<t|\mathcal{X}_n)| \leq C_{BE} \frac{\hat{\E}_n\|W_1\|^3}{\sqrt{n}}+ o_P(1),
$$
where $Z^*_w$ is a normal distribution with variance ${\sf Var}(W^*_1|\mathcal{X}_n)$
and $\hat{\E}_n\|W_1\|^3 = \frac{1}{n}\sum_{i=1}^n W_i^3$.

By the strong law of large number, $\hat{\E}_n\|W_1\|^3<2 \E\|W_1\|^3$ almost surely.
Because ${\sf Var}(W^*_1|\mathcal{X}_n) -{\sf Var}(W_1) =O_P\left(\frac{1}{\sqrt{n}}\right)$
implies $\sup_t|P(Z^*_w<t|\mathcal{X}_n)-P(Z_w<t)| = O_P\left(\frac{1}{\sqrt{n}}\right)$,
we conclude
$$
\sup_{t}\left|P(a^T\Delta^*_n<t|\mathcal{X}_n) - P(a^T\Delta_n<t)\right| = O_P\left(\frac{1}{\sqrt{n}}\right).
$$
Finally, by choose $a$ to be the unit vector along each coordinate, we obtain the desire result for the
bootstrap CIs.

\end{proof}

\begin{proof}[Proof of Theorem~\ref{thm::increase}]

The high level ideas of this proof is very similar to that of the previous theorem.
In the first part, we derive the Berry-Esseen bound of the ELBO estimator.
In the second part, we prove the bootstrap consistency.
Note that in the increasing-dimensional case,
many smaller-order approximations (e.g., those from a Taylor expansion)
may depend on the dimension $d$
and may no longer be small.
So we need to examine each approximation term.


{\bf Part I: Berry-Esseen Bound.}
Recall that $\Psi(\theta) = \E(\cE(\theta|X_1))$
and $\Psi_\theta,\Psi_{\theta\theta}$ are the gradient and Hessian matrix of $\Psi$.
Let
\begin{align*}
\Psi_n(\theta) &= \frac{1}{n}\sum_{i=1}^n \cE(\theta|X_i)\\
\Psi_{\theta,n}(\theta) &= \frac{1}{n}\sum_{i=1}^n \nabla\cE(\theta|X_i)\\
\Psi_{\theta\theta,n}(\theta) &= \frac{1}{n}\sum_{i=1}^n \nabla\nabla \cE(\theta|X_i)
\end{align*}
denote the corresponding empirical versions.

Because $ 0 = \Psi_{\theta,n}(\hat{\theta}_{ELBO}) = \Psi_{\theta}(\theta_{ELBO})$,
by Taylor's theorem
\begin{align*}
\Psi_{\theta,n}(\theta_{ELBO}) - \Psi_{\theta}(\theta_{ELBO})
&= \Psi_{\theta,n}(\theta_{ELBO}) - \Psi_{\theta,n}(\hat{\theta}_{ELBO})\\
& = \Psi_{\theta\theta,n}(\theta_{ELBO}) (\theta_{ELBO}-\hat{\theta}_{ELBO})  + E_{1,n}\\
& = (\Psi_{\theta\theta}(\theta_{ELBO}) + E_{2,n})(\theta_{ELBO}-\hat{\theta}_{ELBO}) + E_{1,n}\\
& = \Psi_{\theta\theta}(\theta_{ELBO})(\theta_{ELBO}-\hat{\theta}_{ELBO}) \\
&\qquad+ E_{2,n} (\theta_{ELBO}-\hat{\theta}_{ELBO}) + E_{1,n},
\end{align*}
where $E_{1,n}\in \R^d$ is a vector about the second order Taylor approximation error and
$E_{2,n}=\Psi_{\theta\theta,n}(\theta_{ELBO}) - \Psi_{\theta\theta}(\theta_{ELBO})$.

Let $Z_n = \Psi_{\theta,n}(\theta_{ELBO}) - \Psi_{\theta}(\theta_{ELBO})$ denotes
the empirical gradient minus the corresponding expected gradient.
By assumption (A1), $\Psi_{\theta\theta}(\theta_{ELBO})$ is always invertible,
so multiplying $\Omega = \Psi^{-1}_{\theta\theta}(\theta_{ELBO})$ in both sides and rearranging the equation lead to
\begin{align*}
\tilde{\Delta}_n = \hat{\theta}_{ELBO}- \theta_{ELBO} = - \Omega Z_n - \Omega E_{2,n} (\theta_{ELBO}-\hat{\theta}_{ELBO})
-\Omega E_{1,n}.
\end{align*}
The first quantity $ \Omega Z_n$ has an asymptotic normality because it contains an empirical sum minus the corresponding expectation.

To derive the asymptotic normality of $\Delta_n$,
we need
\begin{equation}
\begin{aligned}
\|\sqrt{n}a_n^T\tilde{\Delta}_n &+ \sqrt{n}a_n^T\Omega Z_n\| \\
&= \sqrt{n}\|\underbrace{a_n^T\Omega E_{2,n} (\theta_{ELBO}-\hat{\theta}_{ELBO})}_{(I)}-\underbrace{a_n^T\Omega E_{1,n}}_{(II)}\| = o_P(1)
\end{aligned}
\label{eq::pf::normal}
\end{equation}
for any sequence of unit vectors $a_n\in\R^d$.

For part (I),
because $E_{2,n} = \Psi_{\theta\theta,n}(\theta_{ELBO}) - \Psi_{\theta\theta}(\theta_{ELBO})$
is the average of IID random matrices minus the corresponding expectation, the matrix Bernstein inequality
(see, e.g., Theorem 6.2 in \citealt{tropp2012user})
implies $\|E_{2,n}\|_{2} = O_P\left(\sqrt{\frac{\log^2 d}{n}}\right)$, where $\|\cdot\|_{2}$ is the matrix $2$-norm.
This, along with the fact that
assumption (A1) implies $\|\Omega\|_{2}$ being bounded, implies
\begin{align*}
\sqrt{n}\|a_n^T\Omega E_{2,n} &(\theta_{ELBO}-\hat{\theta}_{ELBO})\|\\ &\leq
\sqrt{n} \|a_n\|\|\Omega\|_{2} \underbrace{\|E_{2,n}\|_{2}}_{=O_P\left(\sqrt{\frac{\log^2 d}{n}}\right)} \underbrace{\|\theta_{ELBO}-\hat{\theta}_{ELBO}\|}_{=O_P\left(\sqrt{\frac{d}{n}}\right)} \\
&= O_P\left(\sqrt{\frac{d\log^2 d}{n}}\right).
\end{align*}
This bounds the contribution of (I).

For part (II), we only need to focus on bounding $\|E_{1,n}\|$ because $\|\Omega\|_2$ is bounded.
By the Taylor's theorem, the $\ell$-th element of $E_{1,n}$ can be written as
$$
E_{1,n,\ell} = (\theta_{ELBO}-\hat{\theta}_{ELBO})^T A_\ell (\theta_{ELBO}-\hat{\theta}_{ELBO})
$$
with
$$
A_\ell = \int_{t=0}^{t=1}  \frac{\partial}{\partial \theta_\ell} \Psi_{\theta\theta,n}\left(\hat{\theta}_{ELBO}+t(\theta_{ELBO}-\hat{\theta}_{ELBO})\right) dt.
$$
Let $\hat{\mu} = \frac{\theta_{ELBO}-\hat{\theta}_{ELBO}}{\|\theta_{ELBO}-\hat{\theta}_{ELBO}\|}$ denote the direction of $\theta_{ELBO}-\hat{\theta}_{ELBO}$,
and $r_n = \|\theta_{ELBO}-\hat{\theta}_{ELBO}\|$,
and $e_\ell\in\R^d$ be the unit vector pointing toward the $\ell$-th coordinate.
Then we can rewrite $E_{1,n,\ell}$ as
$$
E_{1,n,\ell} = r_n^2 \int_{t=0}^{t=1} \nabla_{e_\ell}\nabla_{\hat{\mu}}\nabla_{\hat{\mu}}\Psi_n(\hat{\theta}_{ELBO}+t \cdot r_n\cdot \hat{\mu}_n) dt.
$$
Therefore,
$$
E_{1,n} = r_n^2 \int_{t=0}^{t=1} \nabla\nabla_{\hat{\mu}}\nabla_{\hat{\mu}}\Psi_n(\hat{\theta}_{ELBO}+t\cdot r_n\cdot \hat{\mu}_n) dt
$$
and assumption (A2) implies that
\begin{equation}
\begin{aligned}
\|E_{1,n}\| &= r_n^2 \left\|\int_{t=0}^{t=1} \nabla\nabla_{\hat{\mu}}\nabla_{\hat{\mu}}\Psi_n(\hat{\theta}_{ELBO}+t\cdot r_n\cdot \hat{\mu}_n) dt\right\|\\
&\leq r_n^2 c_1\\
& = O_P\left(\frac{d}{n}\right).
\end{aligned}
\label{eq::E1}
\end{equation}
By assumption (A1), $\|\Omega\|_2$ bounded so
$$
\|\sqrt{n}a_n^T\Omega E_{1,n}\| \leq \sqrt{n}\|\Omega\|_2 \|E_{1,n}\| = O_P\left(\sqrt{\frac{d^2}{n}}\right),
$$
which bounds (II).

As a result, the assumption $\frac{d^2}{n}\rightarrow 0$
implies
$$
\sqrt{n}\|a_n^T\Omega E_{2,n} (\theta_{ELBO}-\hat{\theta}_{ELBO})-a_n^T\Omega E_{1,n}\| = o_P(1)
$$
so equation \eqref{eq::pf::normal} holds.

To obtain the Berry-Esseen bound,
after rearranging equation \eqref{eq::pf::normal},
$$
\sqrt{n}a_n^T\tilde{\Delta}_n = -\sqrt{n}a_n^T\Omega Z_n + o_P(1) = \sqrt{n}\bar{W}_n+o_P(1),
$$
where $\bar{W}_n = \frac{1}{n}\sum_{i=1}^n W_i$
and $W_i = -a_n^T \Omega (\Psi_{\theta}(\theta_{ELBO}|X_i) - \Psi_{\theta}(\theta_{ELBO}))$.
Note that the $W_1,\cdots,W_n$ are also IID.
Thus, by assumption (A3) and the Berry-Esseen theorem \citep{berry1941accuracy,esseen1942liapounoff} we conclude that
\begin{equation}
\begin{aligned}
\sup_{t}\left|P(\sqrt{n}a_n^T\tilde{\Delta}_n<t)-P(\sigma(a_n)Z<t)\right|&= o_P(1)+c_{BE}\frac{\E(|a_n^T W_1|^3)}{\sqrt{n}}\\
& = o_P(1) + o(1),
\end{aligned}
\label{eq::BE_increase}
\end{equation}
where $Z\sim N(0,1)$ and $\sigma^2(a_n) = {\sf Var}(W_1) = a_n^T\Omega {\sf Cov}(\Psi_{\theta}(\theta_{ELBO}|X_1))\Omega a_n$.

{\bf Part II: Bootstrap.}
In the bootstrap world, we are sampling from $\hat{P}_n$.
Thus, all the above derivations hold except that everything is conditional on $X_1,\cdots,X_n$
and the expectation is taken over $\hat{P}_n$ instead of $P$.
So the derivation in part I leads to
\begin{equation}
\begin{aligned}
\sup_{t}|P(\sqrt{n}a_n^T\tilde{\Delta}^*_n<t|\mathcal{X}_n)&-P(\hat{\sigma}_n(a_n)Z<t|\mathcal{X}_n)|\\
&= o_P(1)+c_{BE}\frac{\hat{\E}_n(|a_n^T W_1|^3)}{\sqrt{n}}\\
&= o_P(1)+\frac{c_{BE}}{\sqrt{n}} \cdot \frac{1}{n}\sum_{i=1}^n |a_n^T W_i|^3\\
& \leq o_P(1) +\frac{c_{BE}}{\sqrt{n}} \cdot \frac{1}{n}\sum_{i=1}^n\|W_i\|^3\\
&\leq o_P(1) +\frac{c_{BE}}{\sqrt{n}} \cdot \max\{\|W_1\|^3,\cdots,\|W_n\|^3\}\\
& \leq o_P(1) +  O_P\left(\sqrt{\frac{\log d}{n}} \frac{d^{3/2}}{n^{3/2}}\right) = o_P(1),
\end{aligned}
\label{eq::BT_increase}
\end{equation}
where
\begin{align*}
\hat{\sigma}^2_n(a_n) &= a_n^T\hat{\Omega}_n \hat{{\sf Cov}}_n(\Psi_{\theta}(\hat{\theta}_{ELBO}|X_1))\hat{\Omega}_n a_n\\
\hat{\Omega}_n& = \Psi_{\theta\theta,n}^{-1}(\hat{\theta}_{ELBO})\\
\hat{{\sf Cov}}_n(\Psi_{\theta}(\hat{\theta}_{ELBO}|X_1))& = \sum_{i=1}^n \Psi_{\theta}(\hat{\theta}_{ELBO}|X_i) \Psi_{\theta}(\hat{\theta}_{ELBO}|X_i)^T
\end{align*}
are the empirical versions of $\sigma^2(a_n),\Omega$ and ${\sf Cov}(\Psi_{\theta}(\theta_{ELBO}|X_1))$.

By matrix Bernstein inequality \citep{tropp2012user}, the difference
\begin{align*}
|\hat{\sigma}_n(a_n) &- \sigma(a_n)| \\
&\leq \sup_{a:\|a\|=1} |\hat{\sigma}_n(a) - \sigma(a)|\\
&= \|\hat{\Omega}_n \hat{{\sf Cov}}_n(\Psi_{\theta}(\hat{\theta}_{ELBO}|X_1))\hat{\Omega}_n-\Omega {\sf Cov}(\Psi_{\theta}(\theta_{ELBO}|X_1))\Omega\|_{2}\\
& = O_P\left(\|\hat{\Omega}_n-\Omega\|_{2} + \|\hat{{\sf Cov}}_n(\Psi_{\theta}(\hat{\theta}_{ELBO}|X_1))-{\sf Cov}(\Psi_{\theta}(\theta_{ELBO}|X_1))\|_{2}\right)\\
& = O_P\left(\sqrt{\frac{\log^2 d}{n}}\right) = o_P(1).
\end{align*}
Therefore,
$$
\sup_{t}|P(\P(\sigma(a_n)Z<t)-P(\hat{\sigma}_n(a_n)Z<t|\mathcal{X}_n)| = o_P(1).
$$
This, together with equations \eqref{eq::BE_increase} and \eqref{eq::BT_increase},
implies
$$
\sup_{t}|P(\sqrt{n}a_n^T\tilde{\Delta}^*_n<t|\mathcal{X}_n)-P(\sqrt{n}a_n^T\tilde{\Delta}_n<t))| = o_P(1).
$$
Finally, by choose $a_n$ to be the unit vector along each coordinate, we obtain the desire result for the
bootstrap CIs.

\end{proof}

\section{Assumptions in Westling and McCormick (2015)}	\label{sec::B15}

Here we describe the assumptions (B1--5) in the appendix of \cite{westling2015establishing}. 

\begin{itemize}
\item[(B1)] For all $\theta\in \Theta$ and $P_0$-a.e. $x$, ${\sf ELBO}(\theta,\omega|x)$ is uniquely maximized at $\omega= \omega_{\max}(\theta|x)$, 
which is an element of $\Omega$, an open subset of $\R^s$.
\item[(B2)] $\omega_{\max}(\theta|x)$ is a measurable function of $x$ for all $\theta$ and twice continuously differentiable in a
neighborhood of $\theta_{ELBO}$ for $P_0$-a.e. $x$. 
\item[(B3)] ${\sf ELBO}(\theta,\omega|x)$ is twice continuously differentiable in a neighborhood of $\theta_{ELBO}$ and $\omega_{\max}(\theta_{ELBO}|x)$ for $P_0$-a.e. $x$. 
\item[(B4)] There exists $r_1>0$, $s(x)>0$, $b_1(x)$ and $b_2(x)$ such that 
\begin{enumerate}
\item For all $x\in\R^J$ and $\theta\in B(\theta_{ELBO},r_1)$, 
$$
\omega_{\max}(\theta|x) \in B(\omega_{\max}(\theta_{ELBO}|x), s(x)).
$$
\item For all $x\in\R^J$, $\theta_1,\theta_2\in B(\theta_{ELBO},r_1)$ and $\omega_1,\omega_2\in B(\omega_{\max}(\theta_{ELBO}|x), s(x))$, 
$$
|{\sf ELBO}(\theta_1,\omega_1|x) - {\sf ELBO}(\theta_2,\omega_2|x)|\leq b_1(x) \left(\|\theta_1-\theta_2\|+\|\omega_1-\omega_2\|\right). 
$$
\item For all $\theta_1,\theta_2\in B(\theta_{ELBO}, r_1)$, 
$$
\|\omega_{\max}(\theta_1|x)-\omega_{\max}(\theta_2|x)\|\leq b_2(x) \|\theta_1-\theta_2\|.
$$
\item The functions $b_1,b_2\in L_2(P_0)$. 

\end{enumerate}
\item[(B5)] $|\nabla^2 \cE(\theta|x)|\leq \kappa(x)$ for all $\theta$ in a neighborhood of $\theta_{ELBO}$ and $P_0$-a.e. $x$
for an integrable function $\kappa$.

\end{itemize}

\bibliographystyle{abbrvnat}
\bibliography{VI_paper}

\end{document}